\documentclass[11pt]{article}
\usepackage[utf8]{inputenc}
\usepackage[a4paper,width=175mm,top=25mm,bottom=15mm,bindingoffset=6mm]{geometry}
\usepackage{etex}
\usepackage[all]{xy} 
\usepackage{etoolbox} 
\usepackage{lmodern}
\usepackage{array}
\usepackage{xstring}
\usepackage{longtable}
\usepackage[listings]{tcolorbox} 
\tcbuselibrary{breakable}
\tcbuselibrary{skins}
\usepackage[pdftex]{hyperref}
\usepackage{amsmath}
\usepackage{tikz,tikz-cd,stmaryrd}
\usepackage{fancyhdr}
\usepackage{amssymb}
\usepackage{amsthm}
\usepackage{amsmath} 
\usepackage{amsfonts}
\usepackage{amsthm}
\usepackage{pifont}
\usepackage{amsmath}
\usepackage{epsfig}
\usepackage{latexsym}
\usepackage{amssymb}
\usepackage{color}
\usepackage{amscd}
\usepackage{multirow}
\usepackage{graphicx}
\usepackage{lscape}
\usepackage{float}
\usepackage{xcolor}
\usepackage{bm}
\usepackage{tikz-cd}
 \makeatletter
    \let\stdchapter\section
    \renewcommand*\section{%
    \@ifstar{\starchapter}{\@dblarg\nostarchapter}}
    \newcommand*\starchapter[1]{%
        \stdchapter*{#1}
        \thispagestyle{fancy}
        \markboth{\MakeUppercase{#1}}{}
    }
    \def\nostarchapter[#1]#2{%
        \stdchapter[{#1}]{#2}
        \thispagestyle{fancy}
    }
\makeatother
 
\newtheorem{theorem}{Theorem}[section]
\newtheorem*{theorem*}{Theorem}

\newtheorem{proposition}[theorem]{Proposition}
\theoremstyle{definition}
\newtheorem{definition}[theorem]{Definition}
\theoremstyle{corollary}

\theoremstyle{remark}
\newtheorem{remark}[theorem]{Remark}

\theoremstyle{conclusion}

\usepackage{amsmath}
\usepackage{amssymb}
\usepackage{amsfonts}
\usepackage{esvect}
\usepackage{physics}
\usepackage{mathrsfs}
\usepackage{authblk}
\usepackage{geometry}
\usepackage{abstract}
\usepackage{xcolor}

\title{\bf Algebraic approach and exact solutions of superintegrable systems in 2D Darboux spaces}
\author{\large Ian Marquette \footnote{i.marquette@uq.edu.au}, 
Junze Zhang \footnote{junze.zhang@uqconnect.edu.au} and Yao-Zhong Zhang \footnote{yzz@maths.uq.edu.au} }
\affil{School of Mathematics and Physics, The University of Queensland \\ Brisbane, QLD 4072, Australia}

\begin{document}

\maketitle
\begin{abstract}
\noindent Superintegrable systems in 2D Darboux spaces were classified and it was found that there exist 12 distinct classes of superintegrable systems with quadratic integrals of motion (and quadratic symmetry algebras generated by the integrals) in the Darboux spaces. In this paper, we obtain exact solutions via purely algebraic means for the energies of all the 12 existing classes of superintegrable systems in four different 2D Darboux spaces. This is achieved by constructing the deformed oscillator realization and finite-dimensional irreducible representation of the underlying quadratic symmetry algebra generated by quadratic integrals respectively for each of the 12 superintegrable systems. We also introduce generic cubic and quintic algebras, generated respectively by linear and quadratic integrals and linear and cubic integrals, and obtain their Casimir operators and deformed oscillator realizations. As examples of applications, we present three classes of new superintegrable systems with cubic symmetry algebras in 2D Darboux spaces. 
\end{abstract}

\section{Introduction}

Superintegrable systems of different orders have been attracting a large amount of international research activities, see e.g. $\cite{MR1814439}$, $ \cite{MR2337668},$ $\cite{MR3702572}$, $\cite{MR1939624}$, $ \cite{MR2492581}$, $\cite{MR2385271}$ ,  $\cite{MR3119484}$, $\cite{MR1296410} $ and $\cite{MR2105429}$.
This paper is a contribution to the underlying algebraic structures and exact solutions of superintegrable systems in $2$-dimensional (2D) curved spaces. 

Superintegrable systems in 2D spaces with constant or non-constant curvatures have been widely studied by means of separation of variables and St{\"a}ckel transforms $\cite{MR2226333},$ $\cite{MR3988021},$ $\cite{MR2143019}$, $\cite{MR2143027}$, $\cite{MR2023556} $ and $\cite{MR1878980} $. 
The St{\"a}ckel transforms have been widely studied $\cite{MR846389}$, $\cite{MR2023556}$, $\cite{MR2481484}$ provide useful tools in the classification of 2D superintegrable systems. Through the method of the so-called coupling constant metamorphosis,  St{\"a}ckel transforms $\cite{MR761899}$ enable one to establish the relationship between different superintegrable systems: they provide equivalence classes at the level of integrable and superintegrable Hamiltonians. However, even if such Hamiltonians are connected via the St{\"a}ckel transformations, they are distinct as Sturm-Liouville and spectral problem, and their exact solvability (with possibly different boundary conditions) and algebraic solutions need to be investigated separately. For a given superintegrable Hamiltonian which is separable in various coordinates, its solvability would in general depend on the coordinates used in the separation of variables, e.g. it is exactly solvable in one coordinate system but only quasi-exactly solvable in another coordinate system. 

It is well known that symmetry algebra structures play an important role in the analytic analysis of physical systems. In the context of superintegrable models, the underlying symmetry algebra structures are usually polynomial algebras such as quadratic and cubic algebras.
In $\cite{daskaloyannis2000finite},$ $\cite{MR2337668}$, $\cite{MR2226333}$, rank-1 quadratic algebra structures underlying certain 2D superintegrable systems, generated by integrals of motion of the systems, were exploited.
The authors in these references obtained the Casimir operator and deformed oscillator algebra realization of a generic quadratic algebra, and applied the relates to study the energy spectrum of the superintegrable systems.   In $\cite{MR2023556}$ and $\cite{MR3988021} $, examples of rank-1 cubic and quintic algebras in Darboux spaces were given. Higher order or higher rank polynomial algebras generated by integrals and their deformed oscillator algebra realizations were studied in $\cite{MR2492581},\cite{MR2566882}$, $\cite{MR3114205}$, $\cite{MR3331229}$, $\cite{MR3417994}$, $ \cite{MR3465329}$ $, \cite{MR3549641}$, $ \cite{MR3780676}$, $ \cite{MR3797912}$ and $\cite{MR4318599}$. More recently, by extending the Wigner-In{\"o}n{\"u} method of Lie algebra contraction, the authors in $\cite{MR1738603}$, $\cite{MR3116192}$ showed that quadratic algebras from certain second-order superintegrable systems in 2D spaces are contractions of those with general $3$-parameter potentials on $S^2.$ 

Superintegrable systems in 2D $\textit{Darboux spaces}$ were classified in $\cite{MR2023556}$$\cite{MR1878980}$. In 2 dimensions, there exist $4$ possible Darboux spaces with metrics given by  $\cite{MR0396213}$
 \begin{align*}
 {\rm I.} & \qquad d_1(x,y)  = (x+y)\,dx dy \\
 {\rm II.} & \qquad d_2(x,y)   = \left( \frac{\Omega}{(x-y)^2} + \Lambda\right)\, dxdy \\
 {\rm III.} & \qquad d_3(x,y)  = \left(\Omega\exp(-\frac{x+y}{2}) + \Lambda \exp(-x-y)\right)\,dxdy  \\
{\rm IV.}  & \qquad d_4(x,y)   = \frac{\Omega\left(\exp(\frac{x-y}{2}) + \exp(\frac{y-x}{2})\right) + \Lambda}{\exp(\frac{x-y}{2} + \exp(\frac{y-x}{2}))^2}\, dxdy
\end{align*} 
Here $\Omega,\Lambda \in \mathbb{R}$ are constants. According to the classification in $\cite{MR2023556}$$\cite{MR1878980}$, there exist 12 distinct classes of superintegrable systems with non-trivial potentials in the 2D Darboux spaces. In each case, quadratic integrals of motion of the system were determined and were found to form a quadratic algebra. The wave functions and energy spectra of the systems were obtained by means of separation of variables.  Superintegrable systems in Darboux spaces were also studied in $\cite{MR3988021}$ $\cite{MR4088503}$ $\cite{MR4258337}$. It was shown there that free superintegrable systems (i.e. systems without potentials) in 2D and 3D flat conformal spaces are equivalent to systems in 2D and 3D Darboux spaces, respectively. However, as far as we know, finite dimensional representations of the polynomial algebras and algebraic derivations of the energy spectrum of the superintegrable systems have remained an open problem.

In this paper we present a genuine algebraic approach to superintegrable systems in the 2D Darboux spaces.  The purpose of this paper is twofold. One is to give algebraic solutions to the existing 12 distinct classes of superintegrable systems in the four 2D Darboux spaces. This is achieved by constructing the finite dimensional irreducible representation of  the quadratic algebras underlying the 12 superintegrable systems via the deformed oscillator algebra techniques in $\cite{MR1814439}$ and $\cite{MR2492581}$. As one will see, energy spectrum for superintegrable systems in Darboux spaces are often determined by very complicated algebraic equations whose analytic and closed-form solutions can only be obtained by restricting the model parameter spaces. The second purpose is to investigate superintegrable systems in 2D Darboux spaces with linear, quadratic or cubic integrals of motion. It was found in $\cite{MR3988021}$ that the free systems with linear and quadratic integrals in 2D Darboux spaces have cubic algebras as their underlying symmetry algebras. We will introduce generic cubic and quintic algebras, generated by linear and quadratic integrals and linear and cubic integrals, respectively, and construct their Casimir operators and deformed oscillator algebra realizations. We also present three classes of new superintegrable systems with non-trivial potentials in 2D Darboux spaces which have cubic algebras as their symmetry algebras. These superintegrable systems do not seem to belong to the families classified in $\cite{MR2023556}$$\cite{MR1878980}$ for systems with quadratic integrals in 2D Darboux spaces.
 
This paper is organised as follows. In Section $\ref{3}$, we obtain the Casimir operators, the deformed oscillator algebra realizations and finite-dimensional irreducible representations for the quadratic algebras generated by the quadratic integrals of motion of the 12 superintegrable systems in 2D Darboux spaces. This enables us to give an algebraic derivation for the energy spectra of all the 12 classes of superintegrable systems. In Section 3, we introduce generic cubic and quintic algebras generated by linear and higher order integrals of motion. We construct their Casimir operators and deformed oscillator algebra realizations. We also present examples of new superintegrable systems with linear and quadratic integrals in the 2D Darboux spaces. In Section 4, we provide a summary of the main results of our work.

\section{Solutions of the 12 distinct classes of superintegrable systems in 2D Darboux spaces}
\label{3}

Consider a superintegrable system in a 2D Darboux space with coordinates $(x,y)$ and metric $g_{ij}(x,y)$. The Hamiltonian of the system with potential $V(x,y)$ is given by
$$ \hat{\mathcal{H}}   = \sum_{j,k=1}^2 \frac{1}{\sqrt{\det (g_{jk})}} \dfrac{\partial }{\partial x_k} \left( \sqrt{\det (g_{jk})} g_{jk} \dfrac{\partial  }{ \partial x_k} \right)  + V(x,y)  .$$ 
Let $\hat{X}$ be an $\textit{integral of motion}$ (aka, $\textit{constant of motion}$) of the system which commute with the Hamiltonian, i.e. $[\hat{X},\hat{\mathcal{H}}] = 0.$ An integral of motion is said to be a $\textit{polynomial in momenta of degree}$ $p$, denoted by $\deg \hat{X} = p$, if it has the form  \begin{align*}
    \hat{X} = \sum_{j=0}^p r_j(x,y)\, \partial_x^{p-j}\,\partial_y^j +s(x,y),
\end{align*} 
where $r_j(x,y),\, s(x,y)$ are smooth functions in the coordinates $x,y$. In particular, integrals of motion of degree 1, 2 or 3 are usually called linear, quadratic or cubic integrals, respectively. Note that the Hamiltonian has degree 2, i.e. $\deg \hat{\mathcal{H}}=2$.

As mentioned in the Introduction, superintegrable systems in the four 2D Darboux spaces with quadratic integrals of motion were classified in $\cite{MR2023556}$$\cite{MR1878980}$, and 12 distinct classes of potentials which preserve superintegrability were found. In this section, we present algebraic solutions to all the 12 existing superintegrable systems.

Note that in the following we will use the so-called separable coordinates in $\cite{MR2023556}$$\cite{MR1878980}$ for each case. As indicated in $\cite{MR2023556}$$\cite{MR1878980}$, in such coordinates the parameters $\Omega, \Lambda$ in the metrics of the Darboux spaces can be conveniently absorbed into the model parameters of the systems by redefinition so that they do not appear explicitly in the expressions of Hamiltonians and integrals. 

\subsection{Darboux Space I}
According to the classification in $\cite{MR2023556}$$\cite{MR1878980}$, in the Darboux space I, there are two possible superintegrable systems with potentials given by 
$$V_1(x,y) = \frac{b_1(4 x^2 + y^2)}{4x} + \frac{b_2}{x} + \frac{b_3}{xy^2},$$
$$V_2(x,y) = \frac{a_1}{x} + \frac{a_2 y}{x} + \frac{a_3(x^2 + y^2)}{x}, $$
respectively, where $b_i, a_i$ are real constants.

\subsubsection{Potential $V_1 (x,y)$}
 For superintegrable system in Darboux space I with the Hamiltonian $ \hat{\mathcal{H}} = \frac{1}{4 x} \left( \partial_x^2 + \partial_y^2\right) +V_1(x,y)$ associated to  $V_1$, the constants of motion are given by $\cite{MR1878980}$
 \begin{align*}
     A & = \partial_y^2 + \frac{4b_3}{y^2} + b_1 y^2, \\
     B & = y \partial_y \partial_x - x \partial_y^2 + \frac{\partial_x}{2} - \frac{y^2}{4x} \left( \partial_x^2 + \partial_y^2 \right) + \frac{b_1 y^4}{4 x} + \frac{b_2 y^2}{x} + \frac{b_3(4 x^2 + y^2)}{y^2 x}.
\end{align*}  
These integrals satisfy the following quadratic algebra relations 
\begin{align*}
  &  [A,B] = C, \\
  & [A,C] = -8 \hat{\mathcal{H}} A - 16 b_1 B, \\
  & [B,C] = 6 A^2 + 8 \hat{\mathcal{H}} B + 1 6 b_2 A - 2 b_1(3 + 16 b_3) .
\end{align*} 
This is the symmetry algebra of the superintegrable system   The Casimir of this algebra is given by 
$$ K_1 = C^2 - 4 A^3 + 8 \hat{\mathcal{H}} \{A,B\} - 16 b_2 A^2 - 16b_1 B^2 + 4b_1(11 + 16b_3)A .$$
We can show that with the differential realization of $A, B$ the Casimir $K_1$ has the following form in terms of the Hamiltonian $\hat{\mathcal{H}}$,
\begin{align*}
      K_1   = -4(3 + 16 b_3)\hat{\mathcal{H}}^2 + 16 b_1b_2 (3 + 16 b_3).
\end{align*} 
 
In order to obtain the energy spectrum of the system via algebraic means, we now construct realization of the quadratic algebra in terms of the deformed oscillator algebra of the form
\begin{align}
    [\mathcal{N},b^\dagger] = b^\dagger, \quad \text{ } [\mathcal{N},b] = -b, \quad \text{ } b b^\dagger  = \Phi(\mathcal{N} + 1),\quad \text{ } b^\dagger b = \Phi(\mathcal{N}), \label{eq:alg}
\end{align} 
where $\mathcal{N}$ is the number operator and   $\Phi(z)$ is a well-defined real function satisfying
\begin{align}
    \Phi(0) = 0, \quad \text{ } \Phi(z) > 0, \;\forall z >0. \label{eq:consta}
\end{align}  
$\Phi(x)$ is called the $\textit{structure}$ $\textit{function}$ of the deformed oscillator algebra. 

It is non-trivial to obtain such a realization and the corresponding structure function $\Phi(z)$. After a long computation, we find in the present case that 
\begin{align*}
   A   = 4 \sqrt{-b_1} \,(\mathcal{N} + \eta), \quad \text{ } B = \frac{2\hat{\mathcal{H}}}{\sqrt{-b_1}} \, (\mathcal{N} + \eta) + b^\dagger + b 
\end{align*} 
map the quadratic algebra to the deformed oscillator algebra with structure function given by 
\begin{align*}
   \Phi_1^{(I)}(\mathcal{N},\eta)  = & -\frac{1}{16 b_1} \left(-4 (\mathcal{N}+\eta) 16 b_1 b_2+(-b_1)^{3/2} (16 b_3+11)-4 \hat{\mathcal{H}}^2+2 b_1^{3/2} (16 b_3+3) \right.\\
  & \left.+64 \sqrt{-b_1} b_1 (\mathcal{N}+\eta)^3+16 (\mathcal{N}+\eta)^2 (4 b_1 b_2-\hat{\mathcal{H}}^2)  -(16 b_3+3) (2 b_1 \sqrt{-b_1}-4 b_1b_2+\hat{\mathcal{H}}^2)\right).
\end{align*}  Here $\eta$ is a constant to be determined from the constraints on the structure function $\Phi$.

We now obtain the finite-dimensional unitary irreducible representations (unirreps) of the deformed oscillator algebra in the Fock space. Let $\vert z,E \rangle,$ denote the Fock basis states labelled by the eigenvalues $z$ and $E$ of $\mathcal{N}$ and $\hat{\mathcal{H}}$, respectively. Acting the structure function on the Fock states, we find that it is factorized to the following form
\begin{align*}
    \Phi_1^I(z,\eta) = &  \left(z + \eta - \frac{1}{4} \left(2-\sqrt{1-16 b_3}\right)\right) \left(z + \eta -\frac{1}{4} \left(2+\sqrt{1-16 b_3}\right) \right)  \\
  &   \left(z + \eta+\frac{2 b_1 \left(\sqrt{-b_1}-2 b_2\right)+E^2}{4 (-b_1)^{3/2}}\right).
\end{align*}   
For the unirreps to be finite dimensional, we impose the following
constraints on the structure function,
\begin{align}
   \Phi(0,\eta)=0,\quad  \Phi(p+1, \eta) = 0,  \label{eq:constraint}
\end{align}  
where $p$ is a positive integer, $p=0,1,2,\cdots$. These constraints give $(p + 1)$-dimensional unirreps in the Fock space and their solutions give the constant $\eta$ and energy spectrum $E$ of the underlying superintegrable system. 

There are two sets of solutions from the constraints on the structure function:
\begin{align*}
    \eta & = \frac{1}{4} \left(2+ \epsilon \sqrt{1-16 a_2}\right),\\
    E_{im} & = \pm 2\,\sqrt{-1}\;(-b_1)^{3/4}\,\sqrt{p+1-\frac{\epsilon}{4} \sqrt{1-16 b_3}+\frac{ b_2}{\sqrt{-b_1}}} , 
%\text{ } E_2 = \frac{b_1 \sqrt{-\frac{4 b_2}{\sqrt{-b_1}}+\epsilon\sqrt{1-16 b_3}-4 p-4}}{\sqrt[4]{-b_1}} \\
%\Phi_1(z) & = z (p+1 - z)\left( z - \epsilon \frac{\sqrt{1 - 16h}}{2}  \right), 
\end{align*} 
 and 
\begin{align*}
    \eta & = -\frac{2 b_1 \left(\sqrt{-b_1}-2 b_2\right)+E^2}{4 (-b_1)^{3/2}},\\
    E_\epsilon & = \pm 2(-b_1)^{3/4}\sqrt{p+1+\frac{\epsilon}{4} \sqrt{1-16 b_3}-\frac{b_2}{\sqrt{-b_1}}}, 
\end{align*}
where $\epsilon = \pm 1$. The first set of solutions give complex energies which are not physical and thus will be discarded. So the energy spectrum of the system is given by the second set of solutions  which are real for $\epsilon=+1,\, b_1<0,\,b_2 \leq 0,\, b_3<1/16$. The structure function for the corresponding $(p+1)$-dimensional unirreps is 
\begin{align*}
\Phi^{(I)}_{E_{+}}(z) & = z(z-p-1) \left( z -  \frac{2 b_1 \left(\sqrt{-b_1}-2 b_2\right) + E_{+}^2}{4 (-b_1)^{3/2}} - \frac{1}{4} \left(2+ \sqrt{1-16 b_3}\right)\right).
\end{align*}

In the following subsections, we would only give the values of parameter $\eta$ which can lead to real energies $E$.

\subsubsection{Potential $V_2(x,y)$} 

Constants of motion for the superintegrable system in Darboux space I with the Hamiltonian $\hat{\mathcal{H}} = \frac{1}{4 x} \left( \partial_x^2 + \partial_y^2\right) + V_2(x,y)$ corresponding to the potential $V_2$  are given by   $\cite{MR1878980}$ \begin{align*}
    A = y \partial_y \partial_x - x\partial_y^2 &+ \frac{\partial_x}{2} - \frac{y^2}{4x}\left(\partial_x^2 + \partial_y^2 \right) - \frac{2 a_2 y}{x} + \frac{2 a_2(x^2 - y^2)}{x} + \frac{2 a_2 y (x^2 - y^2)}{x}, \\
   & B = \partial_y^2 + 4 a_2 y + 4 a_3y^2.
\end{align*} 
They satisfy the following quadratic algebra relations \begin{align*}
    &[A,B] = C, \\
    &[A,C] = 16 a_2 \hat{\mathcal{H}}  - 16 a_3 B, \\
    &[B,C] = 16 a_3 A +  8 (a_2^2 + 4a_1a_3)- 8 \hat{\mathcal{H}}^2.
\end{align*} 
The Casimir operator of the algebra is given by 
\begin{align*}
    K_2 = C^2 + 16 a_3 A^2 + 16 a_3 B^2  - 32 a_2 \hat{\mathcal{H}} B + 16 \left((a_2^2 + 4 a_1 a_3)-   \hat{\mathcal{H}}^2 \right) A,
\end{align*} 
which in terms of the differential realization of $A, B$ takes the constant value $K_2 = 64 (a_3^2 -a_1 a_2^2).$   

We then determine the realization of above quadratic algebra in terms of the deformed oscillator algebra $\eqref{eq:alg}$
and apply its finite dimensional unirreps to obtain the energy spectrum of the system.  After computations, we find that 
 \begin{align*}
 A =    4\sqrt{-a_3}(\mathcal{N} + \eta),\quad \text{ } B =  \frac{  a_2\hat{\mathcal{H}}}{a_3} + b^\dagger + b.
\end{align*} 
transform the quadratic algebra to the deformed oscillator algebra with structure function
\begin{align*}
  & \Phi_2^{(I)}(\mathcal{N},\eta) =\frac{\left(4 a_1 a_3+a_2^2-\hat{\mathcal{H}}^2\right)^2}{16 a_3^2}-\frac{a_1 a_2^2}{a_3}-\frac{1}{12} (\mathcal{N} + \eta) \left(\frac{24 a_2 \hat{\mathcal{H}}}{\sqrt{-a_3}}+48 a_3\right)+\frac{a_2 \hat{\mathcal{H}}}{\sqrt{-a_3}}+4 a_3 (\mathcal{N} + \eta)^2+a_3.
\end{align*} 
Moreover, the action of this structure function on the Fock states $\vert z, E\rangle$ is factorized as  \begin{align*}
    \Phi_2^{(I)} (z, \eta) = \left(z + \eta- \frac{m_+(E) + 2 a_3}{4 a_3}\right) \left(z + \eta - \frac{m_-(E) + 2 a_3}{4 a_3}\right),
\end{align*} 
where $\eta$ is a constant to be determined and  
$$m_\pm(E) =\frac{a_2 E}{\sqrt{-a_3}} \pm\sqrt{64 a_1^2 a_3^2+4 a_1 \left(a_2^2 (8 a_3-1)-8 a_3 E^2\right)+4 a_2^4-\frac{a_2^2 (8 a_3+1) E^2}{a_3}+4 E^4} .$$

We now impose the constraints $\eqref{eq:constraint}$ to obtain finite-dimensional unirreps of the algebra. We find that for $p=0,1,2,\cdots,$ we have
\vskip.1in 
{\bf Case 1:} $\eta_-(E) = \frac{1}{4a_3}\left(m_-(E) + 2 a_3\right)$ and 
\begin{align}
\sqrt{64 a_1^2 a_3^2+4 a_1 \left(a_2^2 (8 a_3-1)-8 a_3 E^2\right)+4 a_2^4-\frac{a_2^2 (8 a_3+1) E^2}{a_3}+4 E^4} =2a_3(p+1),
\end{align}
which has solutions only for $a_3>0$ and the energy spectrum of the system is given by
 \begin{align*}
        E_{+a_3}  = \pm\frac{1}{\sqrt{8a_3}}\sqrt{\sqrt{128 a_1 a_2^2 a_3^2+a_2^4 (16 a_3+1)+64 a_3^4 (p+1)^2}+32 a_1 a_3^2+a_2^2 (8 a_3+1)},
%\text{ } E_4 = \frac{1}{2 \sqrt{2}}\sqrt{\frac{1}{a_3}\left(n(p)+32 a_1 a_3^2+a_2^2 (8 a_3+1)\right)}.
\end{align*} 
Notice that $E_{+a3}$ is real for $a_1>0,\, a_3>0$.
\vskip.1in
{\bf Case 2}  $\eta_+(E) = \frac{1}{4a_3}\left(m_+(E) + 2 a_3\right)$ and 
\begin{align*}
 \sqrt{64 a_1^2 a_3^2+4 a_1 \left(a_2^2 (8 a_3-1)-8 a_3 E^2\right)+4 a_2^4-\frac{a_2^2 (8 a_3+1) E^2}{a_3}+4 E^4} =-2a_3(p+1),
\end{align*} which has solutions only for $a_3<0$ and the energies of the system are
\begin{align}
E_{-a_3} =   \pm \,\frac{1}{ \sqrt{-8a_3}}\sqrt{\sqrt{128 a_1 a_2^2 a_3^2+a_2^4 (16 a_3+1)+64 a_3^4 (p+1)^2}-32 a_1 a_3^2-a_2^2 (8 a_3+1)}.
\end{align}
Obviously for $a_3<0$ there exist ranges of model parameters $a_1,\,a_2$ such that the eneries $E_{-a_3}$ of the system are real.

The structure function for both cases 1 and 2 corresponding to the $(p+1)$-dimensional unirreps of the algebra is given by $ \Phi^{(I)}_{E_\pm a_3}(z) = z(z-p-1)$. %where $n(p) = \sqrt{128 a_1 a_2^2 a_3^2+a_2^4 (16 a_3+1)+64 a_3^4 (p+1)^2}.$ 

\subsection{Darboux Space II}

In the Darboux space II, there are three superintegrable systems with potentials given by $\cite{MR2023556}$
\begin{align*}
  &  V_1(x,y) = \frac{x^2}{x^2 + 1} \left( a_1 \left( \frac{x^2}{4} + y^2 \right) + a_2 y + \frac{a_3}{x^2} \right), \\
  &  V_2(x,y) = \frac{x^2}{x^2 +1} \left( b_1(x^2 + y^2)  + \frac{b_2}{x^2} + \frac{b_3}{y^2}\right) \\
  & V_3(x,y) = \frac{c_1 + \frac{c_2}{x^2} + \frac{c_3}{y^2}}{x^2 + y^2 + \frac{1}{x^2} + \frac{1}{y^2}},
 \end{align*} 
 respectively, where $a_j,b_j,c_j$ are real constants. 
 
 \subsubsection{Potential $V_1(x,y)$}
 
The constants of motion of the superintegrable system in Darboux space II with the Hamiltonian $ \hat{\mathcal{H}} = \frac{x^2}{x^2 + 1}   \left( \partial_x^2 + \partial_y^2\right) +   V_1(x,y)$ associated to the potential $V_1$ are  \begin{align*}
  &  A = \partial_y^2 + a_1y^2 + a_2y, \\
   B  = \frac{2 y}{x^2 + 1} \left( \partial_y^2 - x^2 \partial_x^2 \right) & + 2x \partial_x \partial_y + \partial_y + \frac{a_1}{2} y \left(x^2 + \frac{x^2 + 4 y^2}{x^2 + 1}\right) + \frac{a_2}{2} \left(x^2 + \frac{ 4 y^2}{x^2 + 1}\right) - \frac{2 a_3 y}{x^2 +1}.
\end{align*} 
They satisfy the following quadratic algebra relations $\cite{MR2023556}$ \begin{align*}
    & [A,B] = C,\\
    & [A,C] =  - 4 a_1 B - 4a_2 A,\\
    & [B,C] = - 24 A^2 + 4 a_2 B + 32 \hat{\mathcal{H}} A - 8 \hat{\mathcal{H}}^2 - 8a_1 \hat{\mathcal{H}} + 6 a_1 + 8 a_1 a_3.
\end{align*}  
Its Casimir operator can be shown to be given by \begin{align*}
    K_1  = C^2 - 16 A^3 + 4 a_1 B^2 + 4 a_2\{A,B\}+ \left(4 a_1( 4 a_3 - 11)  - (16 a_1 \hat{\mathcal{H}}+ 16  \hat{\mathcal{H}}^2 ) \right) A + 32 \hat{\mathcal{H}}A^2.   
\end{align*} 
In term of the differential realization of $A, B$,  the Casimir $K_1$ takes the simple form   $K_1= (32 a_1 +4 a_2^2) \hat{\mathcal{H}} - a_2^2  (3 +4 a_3 )$.

By computations similar to those in the previous subsection, we find that \begin{align*}
    A = 2 \sqrt{-a_1} (\mathcal{N} + \eta), \quad \text{ } B = \frac{2a_2}{\sqrt{-a_1}}(\mathcal{N} + \eta) + \frac{a_2 \hat{\mathcal{H}}}{a_1} + b^\dagger + b  
\end{align*} map the quadratic algebra to the deformed oscillator algebra $\eqref{eq:alg}$ with structure function given by
\begin{align*}
   \Phi_1^{(II)}(\mathcal{N},\eta)= &  -12 a_1^3 - 3 \sqrt{-a_1 } a_1 a_2^2 - 16 a_1^3 a_3 -  4 \sqrt{-a_1} a_1 a_2^2 a_3 + 32 \sqrt{-a_1} a_1^2 \hat{\mathcal{H}} + 16 a_1^3 \hat{\mathcal{H}} \\
   & - 8 a_1 a_2^2 \hat{\mathcal{H}} + 4 \sqrt{-a_1} a_1 a_2^2 \hat{\mathcal{H}} + 16 a_1^2 \mathcal{H}^2 + 4 \sqrt{-a_1} a_2^2 \mathcal{H}^2 \\
   & + (\mathcal{N} + \eta) \left(88 a_1^3 + 16 \sqrt{-a_1} a_1 a_2^2 + 32 a_1^3 a_3   -128 \sqrt{-a_1} a_1^2 \hat{\mathcal{H}}- 32 a_1^3 \hat{\mathcal{H}}+ 16 a_1 a_2^2 \hat{\mathcal{H}} - 
    32 a_1^2 \hat{\mathcal{H}}^2\right) \\
    & +  (\mathcal{N} + \eta)^2 \left(-192 a_1^3 - 16 \sqrt{-a_1} a_1 a_2^2 + 128 \sqrt{-a_1} a_1^2 \hat{\mathcal{H}}\right)  + 128 a_1^3 (\mathcal{N} + \eta)^3.
\end{align*}  Here $\eta$ is a constant to be determined from the constraints of the structure function. Acting  on the Fock basis states $\vert z,E \rangle$, the structure function $\Phi_1^{(II)}$ becomes factorized 
\begin{align*}
  \Phi_1^{(II)}(z,\eta) = & \left(z + \eta -\frac{f_1(E)  +\omega(E) +f_2(E)}{24 a_1^3  } \right)\\
  & \left(z + \eta -\frac{1 }{96 a_1^3} \left(4 f_1(E)-2  \left(1-i\sqrt{3}\right) \omega(E) 
   +\left(1+i \sqrt{3}\right)f_2(E) \right) \right) \\
   & \left(z + \eta -\frac{1}{96 a_1^3}\left(4 f_1(E)-2 \left(1+i \sqrt{3}\right) \omega(E)+\left(1-i \sqrt{3}\right)f_2(E)\right) \right),
\end{align*}  where \begin{align*}
    f_1(E) = & a_1\left(12 a_1^2+\sqrt{-a_1} a_2^2   + 8   (-a_1)^{3/2}E \right) , \\
    f_2(E) = & \frac{1}{\omega(E)}\left(a_1^3(12 a_1^3(4 E - 4a_3 + 1)-  a_2^4-16 a_1^2 E^2 -8 a_1  a_2^2 E)\right), \\
    \omega(E)  = & \sqrt[3]{\tau_1(E)+\tau_2(E)},\\ 
    \tau_1(E) = & 6 a_1^6\left(a_2^4+8 a_1 E a_2^2+16 a_1^2 E^2+a_1^3 (-16 a_3+16 E+4)\right) \, \sqrt{ 3(4 a_3-4 E-1)  }, \\
     \tau_2 (E)=& a_1a_2^6 (-a_1)^{7/2} +12 a_2^4 E (-a_1)^{11/2}-48 a_2^2 E^2 (-a_1)^{13/2}\\
     & -4 \left(9 a_2^2 (4 a_3-4 E-1)-16 E^3\right) (-a_1)^{15/2} 
       -144 E (-4 a_3+4 E+1) (-a_1)^{17/2}.
\end{align*}    

To determine the constant $\eta$ and energy spectrum $E$ of the superintegrable system, we  impose the constraints $\eqref{eq:constraint}$ which give $(p+1)$-dimensional unirreps of the algebra. We find 
\vskip.1in
{\bf Case 1:}  The constant $\eta$ is given by $$\eta_1(E) =\frac{f_1(E)  +\omega(E) +f_2(E)}{24 a_1^3  }$$ and the energy $E$ satisfies the algebraic equation, 
\begin{align}
  \omega(E) +f_2(E) +  \frac{1}{2}\left(1-\epsilon\,i \sqrt{3}\right) \omega(E)-\frac{1}{4}\left(1 + \epsilon \, i \sqrt{3}\right)f_2(E)  =-  24  (p+ 1) a_1^3. \label{eq:eq1}
\end{align} 

\vskip.1in
{\bf Case 2:}  $$\eta_2(E) =\frac{1 }{96 a_1^3}\left(4 f_1(E) - 2 \left(1-\epsilon\,i \sqrt{3}\right) \omega(E)+\left(1+ \epsilon \, i \sqrt{3}\right)f_2(E)\right)$$ and the energy is determined by 
\begin{align}
    \omega(E) +f_2(E) +  \frac{1}{2}\left(1-\epsilon\,i \sqrt{3}\right) \omega(E)-\frac{1}{4}\left(1 + \epsilon \,  \sqrt{3}i\right)f_2(E)  =  24  (p+ 1) a_1^3. \label{eq:eq2}
\end{align} 
In both cases above, $\epsilon=\pm 1$.

The energy spectrum $E$ of the system are obtained by solving the algebraic equations $\eqref{eq:eq1}$ and $\eqref{eq:eq2}$. However, it is in general very difficult to obtain analytical solutions of these equations,  due to their complicated form. To demonstrate that these equations have real solutions, we have a closer look at restricted model parameter spaces. Without the loss of generality, we consider the case where $-a_1=a_2=a_3=a$ for any $a\in\mathbb{R}$. For such model parameters, the structure function has the simple form, 
\begin{align*}
    \Phi_1^{(II)}(z,\eta) = &\left(z + \eta - \frac{1}{8} \left(\sqrt{a}+4\right) \right) \left(z + \eta - \frac{1 }{4 a} \left( 2 \sqrt{a} E+2 a-a\sqrt{4 E+1-4 a}\right)\right) \\
    & \left(z + \eta - \frac{1}{4 a}\left(2 \sqrt{a} E+2 a+a\sqrt{4 E+1-4 a}\right) \right). 
\end{align*} 
Imposing the constraints $\eqref{eq:constraint}$ on the structure function lead to the determination of constant $\eta$ and energy $E$ of the superintegrable system for the model parameters $-a_1=a_2=a_3=a$. There are two sets of solutions: 
One is that $\eta =\frac{1}{8} \left(\sqrt{a}+4\right) $ and \begin{align*}
    E_\epsilon  = \frac{1}{4} \left(8 \sqrt{a} (p+1)+3 a+  2 \epsilon\sqrt{8 a^{3/2} (p+1)-2 a^2+a} \right),
\end{align*} 
where $\epsilon = \pm 1$, with the associated structure function $\Phi^{(II)}_{E_\epsilon}(z) = z(p+1 - z)^2.$
The energy spectrum $E_\epsilon$ is real for $0 < a \leq 1/2 .$  

The second set of solutions is given by
$$\eta (E) =\frac{1 }{4 a} \left( 2 \sqrt{a} E+2 a-  a\sqrt{    4 E+1-4 a}\right) $$ and the corresponding energy spectrum of the system and  structure function for the $(p+1)$-dimensional unirreps of the deformed oscillator algebra are given by 
\begin{align*}   
& E = p(p+2)+a+\frac{3}{4} ,\\
&\Phi^{(II)}_E (z) = z(z-p-1)\left(z+  \frac{1}{8 a}\left(3 a^{3/2}- 4 a (p+1) +\sqrt{a} (4 p^2+8 p+3)\right) \right). 
\end{align*} 

Thus we have demonstrated that there exist indeed non-trivial model parameters which give real energies of the superintegrable system in both Case 1 and Case 2 above.

 \subsubsection{  Potential $  V_2(x,y) $}
 
The superintegrable system in Darboux II with potential $V_2(x,y)$ has Hamiltonian $\hat{\mathcal{H}} =\frac{x^2}{x^2 + 1}   \left( \partial_x^2 + \partial_y^2\right) +V_2(x,y)$. This system possesses the following integrals of motion $\cite{MR2023556}$ 
\begin{align*}
& A = \partial_y^2 + b_1 y^2 + \frac{b_3}{y^2},\\
   B = \frac{ (y^2 - x^4) \partial_y^2 + x^2 (1 - y^2) \partial_x^2}{x^2 + 1} &  + 2xy \partial_x\partial_y + x \partial_x + y \partial_y - \frac{1}{4} + \frac{x^2 + y^2}{x^2 +1} \left( b_1(x^2 +y^2)-b_2 -b_3\frac{x^2}{y^2} \right),
\end{align*} which form the quadratic algebra relations \begin{align*}
    & [A,B] = C, \\
    & [A,C] = 8 A^2 - 16 b_1 B + 16 b_1 \hat{\mathcal{H}}-16 b_1 (b_2 + b_3 + \frac{3}{4}),  \\
    & [B,C] = - 8 \{A,B\} + 8 \hat{\mathcal{H}} B + 12 A -8 \hat{\mathcal{H}}^2 + 8 (b_2 - b_3 - \frac{3}{4}) \hat{\mathcal{H}}. 
\end{align*} 
By a direct calculation, we find the Casimir operator of the algebra \begin{align*}
    K_2 = & C^2 - 8 \{A^2,B\} + 8 \hat{\mathcal{H}}\{A,B\}  + 16 b_1 B^2 + 76 A^2 + \left(16(b_3 - b_2 + \frac{19}{4}) \hat{\mathcal{H}}- 16 \hat{\mathcal{H}}^2 \right) A\\
    &+  \left(8b_1 (4(b_2 + b_3)+3)- 32 b_1 \hat{\mathcal{H}}\right)B .
\end{align*} 
This Casimir operator can be expressed in terms of Hamiltonian as 
$$K_2 =  -16 \left(b_1 + b_3 + \frac{3}{4}\right) \hat{\mathcal{H}}^2 - 8 b_1 (4b_3 -4b_2 + 3) \hat{\mathcal{H}} + b_1 \left( 36 + 48 b_3 - (4b_3 -4b_2 + 3)^2 \right).$$
 
It can be shown that after the change of basis \begin{align*}
    &A  = 4\sqrt{-b_1} (\mathcal{N} + \eta),\\
    B = 8 (\mathcal{N} + \eta)^2 - \frac{2 \hat{\mathcal{H}}}{\sqrt{-b_1}} & - \frac{16  (b_2 + b_3 + \frac{3}{4}) (\mathcal{N} + \eta)-   b_1 \hat{\mathcal{H}}}{b_1} + b^\dagger + b,
\end{align*} the quadratic algebra becomes the deformed oscillator algebra with structure function 
\begin{align*}
\Phi_2^{(II)}(\mathcal{N},\eta) =& \frac{1}{16} \left(4 b_3+16 \mathcal{N}^2+16 \mathcal{N} (2 \eta-1)+16 \eta^2 - 16\eta+3\right) \\
& \left(4  b_2+\frac{1-4 \hat{\mathcal{H}} \left(\sqrt{-b_1}+2\mathcal{N}+2\eta-1\right)}{\sqrt{-b_1}}+16 \mathcal{N}^2+32 \mathcal{N} \eta-16 \mathcal{N}+16 \eta^2 - 16\eta+3 \right).
\end{align*} 
On the Fock states $\vert z,E \rangle,$ the structure function is factorized as follows
\begin{align*}
   \Phi_2^{(II)}(z,\eta)   =& \left( z + \eta - \frac{1}{4} \left(2-\sqrt{1-4 b_3}\right)\right)\left( z + \eta -\frac{1}{4} \left(2+\sqrt{1-4 b_3}\right) \right)  \\
   & \left( z + \eta - \frac{ 2 b_1-\gamma_+(E)}{4 b_1}\right)  \left( z + \eta - \frac{2b_1-\gamma_-(E)}{4 b_1}\right),
\end{align*} 
where $$\gamma_\pm(E)  = \sqrt{b_1^2 (4 E-4 b_2+1)}\pm  \sqrt{-b_1} E.$$ 

Imposing the constraints  $\eqref{eq:constraint},$ for any $p \in \mathbb{N}^+$ we get the following values for the parameter $\eta$ and energy $E$:
\vskip.1in
{\bf Case 1.}  $\eta_+ (E) = \frac{1}{4b_1}\left(2 b_1-\gamma_+(E)\right)$. This $\eta$ value gives the following energy spectrum of the system and the corresponding structure function of the deformed oscillator algebra
$$E_-= - (p+2)\sqrt{-b_1},$$ 
$$\Phi^{(II)}_{E_-}(z)=z(z-p-1)\left(z+\frac{1}{4b_1}\left(b_1\sqrt{1-4b_3}-\gamma_+(E_-)\right)\right) \left(z-\frac{1}{4b_1}\left(b_1\sqrt{1-4b_3}+\gamma_+(E_-)\right)\right).
$$
The enegry $E_-$ is real for $b_1<0$. 
\vskip.1in
{\bf Case 2.} $\eta_- (E) = \frac{1}{4b_1}\left(2 b_1-\gamma_-(E)\right)$. This $\eta$ value gives two sets of energies of the system,
\begin{align}
E_+= & (p+2)\sqrt{-b_1},\label{eta-1}\\
E_\epsilon= & -2b_1+4\sqrt{-b_1}\left(p+1+\frac{\epsilon}{4}\sqrt{1-4b_3}\right)\nonumber\\
   & \pm\sqrt{4b_1^2-16b_1\sqrt{-b_1}\left(p+1+\frac{\epsilon}{4}\sqrt{1-4b_3}\right)+4b_1b_2-b_1}.\label{eta-2}
\end{align}
$E_\epsilon$ above is obtained by solving the algebraic equation  $\gamma_-(E)=4b_1\left(p+1+\frac{\epsilon}{4}\sqrt{1-4b_3}\right)$ from the constraints. (Notice that the other algebraic equation  $\gamma_+(E)=4b_1\left(p+1+\frac{\epsilon}{4}\sqrt{1-4b_3}\right)$ lead to complex solutions and its solutions are not shown here.) Obviously $E_+$ is real for $b_1<0$ and $E_\epsilon$ is real for $\epsilon=+1, b_1<0,\, b_2<\frac{1}{4},\,b_3<\frac{1}{4}$.  The structure functions corresponding to  $E_+, E_\epsilon$ in the Case 2 above are given by 
$$\Phi^{(II)}_{E_+}(z)=z(z-p-1)\left(z+\frac{1}{4b_1}\left(b_1\sqrt{1-4b_3}-\gamma_-(E_+)\right)\right) \left(z-\frac{1}{4b_1}\left(b_1\sqrt{1-4b_3}+\gamma_-(E_+)\right)\right),
$$
$$ \Phi^{(II)}_{E_\epsilon}(z) = z(z-p-1)\left( z + \frac{1}{2b_1}\sqrt{-b_1}\,E_\epsilon\right)   \left(z-\frac{1}{4b_1}\left(\gamma_-(E_\epsilon) +\epsilon b_1\sqrt{1-4 b_3}\right) \right), $$ 
respectively.

\vskip.1in
{\bf Case 3:} $\eta  =\frac{1}{4} \left(2+ \epsilon \sqrt{1-4 b_3}\right) $. The corresponding energies are given the same expression as $E_\epsilon$ above (and are obtained from solving the algebraic equation $ \gamma_+(E)=-4b_1\left(p+1+\frac{\epsilon}{4}\sqrt{1-4b_3}\right)$).

\subsubsection{Potential $V_3(x,y)$}

The constants of motion for the superintegrable system in Darboux space II with the Hamiltonian $\hat{\mathcal{H}} =\frac{x^2}{x^2 + 1}   \left( \partial_x^2 + \partial_y^2\right) +   V_3(x,y)$ associated to the potential $V_3$ are given by \begin{align*}
    A = & \frac{\left(y^2 + \frac{1}{y^2} \right) \partial_x^2 - \left(x^2 + \frac{1}{x^2} \right) \partial_y^2 }{x^2 + y^2 + \frac{1}{x^2} + \frac{1}{y^2}} + \frac{c_1 x^2(y^4 + 1) + c_2(y^4 +1) -c_3(x^4 + 1)}{(x^2y^2 +1)(x^2 + y^2)}, \\
    B = &\frac{c_1 (x^2 + y^2) -c_2(y^4 -1) -c_3(x^4 - 1)}{4(x^2y^2 +1) } + xy(x^2 -y^2) \left(xy \partial_x^2 -xy \partial_y^2 + (x^2 -y^2) \partial_x\partial_y \right) \\
   & + \frac{1}{x^2y^2 + 1}   
   \left[\left(\frac{x^2 -y^2 }{4} + y^4 \right) x^2 \partial_x^2 +\left(\frac{x^2 -y^2}{4} + x^4 \right) y^2 \partial_y^2 + 2 xy \left(\frac{x^2 -y^2}{2} - x^2y^2 \right) \partial_x \partial_y \right] .
\end{align*} 
They form the following  quadratic algebra relations  
\begin{align*}
  &  [A,B] = C, \\
  & [A,C] = 2 A^2 +2 c_1 A + 16 \hat{\mathcal{H}} B + 6 \hat{\mathcal{H}} - 8 \hat{\mathcal{H}}^2, \\
  &   [B,C] = - 2 \{A,B\} + (c_2 + c_3) A - c_1 c_3. 
\end{align*} 
The Casimir operator of this algebra is  \begin{align*}
    K_3 = C^2 - 2 \{A^2,B\} - 16 \hat{\mathcal{H}} B^2 + (c_2 + c_3 +4) A^2 + 2 c_1 \{A,B\}_a -2c_1(c_3 + 2)A + ( 16 \hat{\mathcal{H}}^2 - 12 \hat{\mathcal{H}})B.
\end{align*} 
With the differential realization of $A,B$, the Casimir operator can be expressed in terms of the Hamiltonian as \begin{align*}
    K_3 =  4(c_2 + c_3 ) \hat{\mathcal{H}}^2+ (c_1^2 - 4 c_2 c_3 - 3(c_2 + c_3))\hat{\mathcal{H}}-\frac{3 + 4 c_3}{4}c_1^2.
\end{align*}  

We can convert the quadratic algebra into the deformed oscillator algebra by using the realization \begin{align*} 
 A = 4 \sqrt{\hat{\mathcal{H}}} (\mathcal{N} + \eta), \quad\text{ }B= -2 (\mathcal{N} + \eta)^2 + \frac{c_1}{2\sqrt{\hat{\mathcal{H}}}} (\mathcal{N} + \eta) - \frac{3 \hat{\mathcal{H}} - 4 \hat{\mathcal{H}}^2}{ 8 \hat{\mathcal{H}}} + b^\dagger + b 
\end{align*} with the corresponding structure function given by
\begin{align*}
 \Phi_3^{(II)}&(\mathcal{N},\eta) =-\frac{1}{256 \hat{\mathcal{H}}}\left( 4 c_3-4 \hat{\mathcal{H}}+16\mathcal{N}^2+32 \mathcal{N} \eta-16 \mathcal{N}+16 \eta^2-16 \eta+3\right)  \\
 & \quad\times\left(-c_1^2+4 c_1 \sqrt{\hat{\mathcal{H}}} (2\mathcal{N}+2\eta-1)+\hat{\mathcal{H}} \left(-4 c_2+4 \hat{\mathcal{H}}-16\mathcal{N}^2-32 \mathcal{N} \eta+16 \mathcal{N}-16 \eta^2+16 \eta-3\right)\right). 
\end{align*} 
 By acting $\Phi_3^{(II)}$ on the Fock basis states $\vert z,E \rangle,$ we find that the structure function is factorised as  \begin{align*}
  \Phi_3^{(II)}(z,\eta) = &\left(z + \eta -\frac{1}{4} \left(2-\sqrt{-4c_3 +4 E+1}\right) \right)\left(z +\eta -\frac{1}{4} \left(2+\sqrt{-4c_3 +4 E+1}\right)\right)  \\
 & \times \left(z + \eta -\frac{1}{4 E}\left(-E\sqrt{-4c_2   +4 E +1}+c_1 \sqrt{E}+2 E\right)\right)\\
 & \times \left(z + \eta -\frac{1 }{4 E}\left( E\sqrt{-4c_2   +4 E +1}+ c_1 \sqrt{E}+2 E\right)\right).
\end{align*} 

We now obtain the energy spectrum of the system from the finite-dimensional unirreps of the deformed oscillator algebra. Imposing the constraints $\eqref{eq:constraint}$ which give $(p+1)$-dimensional unirreps for any $p \in \mathbb{N}^+$, we determine the parameter $\eta$ and the energy $E$ of the system. There are two sets of solutions;
\vskip.1in
{\bf Case 1:}  $\eta  (E) =\frac{1}{4} \left(2 - \sqrt{-4c_3 +4 E+1}\right)  $  and the energies are determined by either \begin{align}
   \sqrt{-4c_3 +4 E+1}  - 2(p+1) = 0 \label{eq:333} 
\end{align}
or 
\begin{align}
 \frac{c_1}{\sqrt{E}} +  \sqrt{-4 c_3+4 E+1}+\sqrt{-4 c_2+4 E+1} = 4(p+1). \label{eq:666}
\end{align}
Solution to the algebraic equation (\ref{eq:333}) gives the energies
$$ E_{c3} =p(p+2)+ c_3+\frac{3}{4}. $$
The structure function of the corresponding $(p+1)$-dimensional unirreps is
\begin{align*}
\Phi^{(II)}_{E_{c3}}(z)=& z(z-p-1)\left(z-\frac{1}{2}\left(p+1-\sqrt{(p+1)^2+c_3-c_2}\right)-\frac{c_1}{4\;\sqrt{\left(p+\frac{1}{2}\right)\left(p+\frac{3}{2}\right)+c_3}}\right)\\
&\qquad\left(z-\frac{1}{2}\left(p+1+\sqrt{(p+1)^2+c_3-c_2}\right)-\frac{c_1}{4\;\sqrt{\left(p+\frac{1}{2}\right)\left(p+\frac{3}{2}\right)+c_3}}\right).
\end{align*}
 Other possible energies of the system are given by solutions to the algebraic equation (\ref{eq:666}), which read
 \begin{align*}
    E_{\pm} & = \frac{1}{4}\frac{(p+1+c_1)^2\left(2c_2+2 c_3-1\pm\sqrt{ (p+1+c_1)^2 +4c_2 c_3- (c_2+c_3)+\frac{1}{4}}\,\right)}{(p+1+c_1)^2-(c_2-c_3)^2}.
%    E_3 & = \frac{-a(c_1,c_2,c_3,p)+b(c_1,c_2,c_3,p)}{4 \left(c_1^2+2 c_1 (p+1)-c_2^2+2 c_2 c_3-c_3^2+p^2+2 p+1\right)},
\end{align*}
These energies are real for the model parameters satisfying $4c_2c_3+\frac{1}{4}> c_2+c_3$. The corresponding structure functions for the (p+1)-dimensional unirreps of the algebra are 
\begin{align*}
\Phi^{(II)}_{E_\pm}(z)=& z(z-p-1)\left(z-\frac{1}{2}\sqrt{-4c_3+4E_\pm+1}\right)\\
 & \left(z-\frac{1}{4}\left(\sqrt{-4c_3+4E_\pm+1}-\sqrt{-4c_2+4E_\pm+1}+\frac{c_1}{\sqrt{E_\pm}}\right)\right).
 \end{align*}

\vskip.1in
{\bf Case 2:}  $\eta (E) =  \frac{1}{4E} \left(c_1\sqrt{E}+2 E- E\sqrt{-4c_2   +4 E +1}\right)$. This $\eta$ value gives the following energy spectrum of the system and the corresponding structure function of the unirreps,  
\begin{align*}
&\qquad E_{c2}= p(p+2)+c_2+\frac{3}{4},\\
\Phi^{(II)}_{E_{c2}}(z)=& z(z-p-1)\left(z+\frac{1}{2}\left(p+1-\sqrt{(p+1)^2+c_2-c_3}\right)+\frac{c_1}{4\;\sqrt{\left(p+\frac{1}{2}\right)\left(p+\frac{3}{2}\right)+c_2}}\right)\\
&\qquad\left(z-\frac{1}{2}\left(p+1+\sqrt{(p+1)^2+c_2-c_3}\right)+\frac{c_1}{4\;\sqrt{\left(p+\frac{1}{2}\right)\left(p+\frac{3}{2}\right)+c_2}}\right).
\end{align*}

 \subsection{Darboux Space III}

In Darboux space III, there exist $4$ different potentials. In terms of the separable coordinates $(u, v)$ and $(\mu,\nu)$, they are given by 
\begin{align*}
   & V_1(u,v) = \frac{a_1 u + a_2 v + a_3}{4 + u^2 + v^2},\\
   & V_2(u,v)= \frac{\frac{b_1}{u^2} + \frac{b_2}{v^2} + b_3}{4 + u^2 + v^2}, \\
   & V_3(\mu,\nu) = \frac{c_1(\mu + \nu) + c_2 \frac{\mu+\nu}{\mu\nu} + c_3 \frac{\nu^2 - \mu^2}{\nu^2 \mu^2}}{(\mu + \nu)(2 + \mu - \nu)},\\
   &V_4(\mu,\nu) = \frac{d_1 \mu  + d_2 \nu  + d_3 \nu^2+\mu^2 }{(\mu + \nu)(2 + \mu - \nu)},
\end{align*}
where $a_i, b_i, c_i, d_i$ are real constants.

\subsubsection{Potential $V_1(u,v)$}

The constants of motion of the superintegrable system in Darboux space III with the Hamiltonian $\hat{\mathcal{H}} = \frac{\exp(2u)}{4(\exp(u)) + 1}  \left( \partial_u^2 + \partial_v^2 \right) +  V_1(u,v)$ associated to the potential $V_2$ are given by $\cite{MR2023556}$ \begin{align*}
& A = \frac{(2 + v^2) \partial_u^2 - (2 + u^2) \partial_v^2}{2(4 + u^2 + v^2)} + \frac{a_1 u(2 + v^2) - 2 a_2 v(2 + u^2) + a_3(v^2 - u^2)}{4(4 + u^2 + v^2)},\\
& B = 2 u v \frac{(\partial_u^2 + \partial_v^2)}{2(4 + u^2 + v^2)} - 2 \partial_u \partial_v + \frac{a_1 v ( v^2 - u^2  + 4)+ a_2 u( u^2 - v^2 + 4) -2 a_3 v u}{4(4 + u^2 + v^2)}.
\end{align*}  They form the quadratic algebra with the commutation relations \begin{align*}
    & [A,B] = C, \qquad [A,C] = \hat{\mathcal{H}} B - \frac{a_2a_1}{8}, \qquad \text{ } [B,C] = -\hat{\mathcal{H}} A - \frac{a_2^2 -a_1^2}{16},
\end{align*} which is the symmetry algebra of the superintegrable system.

By a direct computation, we obtain the Casimir operator of this algebra
\begin{align*}
    K_1 = - \hat{\mathcal{H}} A^2 - \hat{\mathcal{H}}B^2 - \frac{a_2^2 - a_1^2}{8} A + \frac{a_1 a_2}{4}B.
\end{align*} We can show that in terms of the Hamiltonian this Casimir operator takes the form  \begin{align*}
    K_1 = -\hat{\mathcal{H}}^3 + \frac{1}{2} (a_3 + \frac{1}{2}) \hat{\mathcal{H}}^2 + \frac{1}{16} (2 a_1^2 + 2 a_2^2 - a_3^2 )\hat{\mathcal{H}} - \frac{a_3(a_1^2 + a_2^2)}{32}.
\end{align*} 

To determine the energy spectrum of the system, we now construct the deformed oscillator algebra realization of the quadratic algebra.   We find  that  \begin{align*}
    A = \sqrt{\hat{\mathcal{H}}} (\mathcal{N} + \eta), \qquad \text{ } B = \frac{a_1a_2}{8 \hat{\mathcal{H}}} + b^\dagger + b ,
\end{align*} trsansform the quadratic algebra into the deformed oscillator algebra with the structure function \begin{align*}
    \Phi_1^{(III)}(\mathcal{N},\eta) = & \frac{1}{256 \hat{\mathcal{H}}}\left(a_1^2+2 a_3 \hat{\mathcal{H}}-4 \hat{\mathcal{H}}^{3/2} \left(2 \sqrt{\hat{\mathcal{H}}}+2\mathcal{N}+2 \eta-1\right) \right) \\
   & \quad\times \left(a_2^2+2 a_3 \hat{\mathcal{H}}+4 \hat{\mathcal{H}}^{3/2} \left(-2 \sqrt{\hat{\mathcal{H}}}+2\mathcal{N}+ 2\eta-1\right)\right).
\end{align*} Here $\eta $ is a constant parameter to be determined from the constraints $\eqref{eq:alg}$. Acting on the Fock basis states $\vert z,E \rangle,$ the structure function $\Phi_1^{(III)}$ becomes
\begin{align*}
   \Phi_1^{(III)}(z,\eta) = &\left(z + \eta -\frac{1}{8 E^{3/2}}\left(a_1^2+2 E \left(a_3-4 E+2 \sqrt{E}\right)\right)\right) \\
   &\quad\times \left(z + \eta +\frac{1}{8 E^{3/2}}\left(a_2^2+2E( a_3 -4E -2 \sqrt{E})\right)\right).
\end{align*}   
The constraints $\eqref{eq:constraint}$ give the $(p+1)$-dimensional unirreps of the deformed oscillaor algebra and their solutions determine the constant $\eta$ and energy spectrum of the superintegrable system. There are two sets of solutions:
\vskip.1in
{\bf Case 1:} $\eta(E) =\frac{1}{8 E^{3/2}}\left(a_1^2+2 E \left(a_3-4 E+ 2  \sqrt{E}\right) \right)$ and energies $E$ are determined by the algebraic equation
\begin{align}
 8 E^{3/2} \left(p+1\right)+4a_3 E+ a_1^2+a_2^2= 16E^2. \label{eq:444}
 \end{align}

 {\bf Case 2:} $\eta(E) = -\frac{1}{8 E^{3/2}}\left(a_2^2+2E\left( a_3 -4E -2 \sqrt{E})\right)\right)$ and energies $E$ satisfy 
\begin{align}
 16E^2+8 E^{3/2} \left(p+1\right)=4a_3 E+a_1^2+a_2^2.\label{eq:555}
\end{align} 
The algebraic equations $\eqref{eq:444}$ and $\eqref{eq:555}$ can be solved by using symbolic computation packages. 

It can be shown that there exist model parameters $a_i$ such that solutions to these algebraic equations for energies are real. 
To demonstrate this, we consider the case in which the model parameters satisfy $a_1 = 0$ and $a_2 =a_3 = 1$. In this case we find that the structure function reduces to
\begin{align*} \Phi_1^{(III)} (z,\eta) = \left( z + \eta - \left(\frac{1}{2}- \sqrt{E}+\frac{1}{4\sqrt{E}}\right)  \right) \left( z+ \eta-   \left(\frac{1}{2}+\frac{1}{8 E^{3/2}} (8 E^2-2 E-1)\right)\right).
\end{align*} 
Imposing the constraints $\eqref{eq:constraint}$ gives the constant $\eta$ and energies as follows. 
\vskip.1in
{\bf a.}  $ \eta(E) = \frac{1}{2} - \sqrt{E}+\frac{1}{4\sqrt{E}} $ which leads to the algebraic equation $ 4 E (1-4 E)+1 +8 E^{3/2}(p+1)= 0  $ for $E$.  This equation has real solution given by 
\begin{align*}
   E_\pm= \frac{1}{48} \left(3 p^2+\sqrt{3}\, g(p)+6 p+9\pm \sqrt{6\,f(p)}\right),
\end{align*} where \begin{align*}
 e(p) = &\sqrt[3]{27 p^4+108 p^3+252 p^2+3 \sqrt{3} \sqrt{(p+1)^4 \left(27 p^4+108 p^3+310 p^2+404 p+575\right)}+288 p+367}  \\
 g(p)  = &\sqrt{3 \left(p^2+2 p+3\right)^2+2\times 2^{2/3} e(p)+\frac{1}{e(p)}\,4 \sqrt[3]{2} \left(6 p^2+12 p+31\right)+8}\\
f(p)  = & ~3 \left(p^2+2 p+3\right)^2-2^{2/3} e(p)-\frac{1}{e(p)}2 \sqrt[3]{2} \left(6 p^2+12 p+31\right)\\
& \qquad +\frac{1}{g(p)}\,3 \sqrt{3} (p+1)^2 \left(p^4+4 p^3+12 p^2+16 p+23\right)+8.
\end{align*} It is clear that $g(p)$ is real for all $p \in \mathbb{N}^+.$  We now show that $f(p)>0$ for all $p \in \mathbb{N}^+.$ Let \begin{align*}
    f_0(p) = 3 \left(p^2+2 p+3\right)^2-2^{2/3} e(p)-\frac{1}{e(p)}\,2 \sqrt[3]{2} \left(6 p^2+12 p+31\right)+8.
\end{align*} By using symbolic computation package, we found that $\frac{df_0(p)}{dp}>0$ for all $p \in \mathbb{N}^+.$ Hence $f_0(p)$ is strictly increasing. Moreover, $ f_0(0) = -62 \sqrt[3]{\frac{2}{15 \sqrt{69}+367}}-2^{2/3} \sqrt[3]{15 \sqrt{69}+367}+35 \cong 12.5741>0$. It follows that $f(p)>0$ for all $p \in \mathbb{N}^+$ and the energy $E$ given above is real. 

\vskip.1in
{\bf b.} $\eta(E) = \frac{1}{2}+\frac{1}{8 E^{3/2}} (8 E^2-2 E-1) $. This leads to the algebraic equation $ 4 E (4 E-1)-1  +8 E^{3/2}(p+1) =0 $. It gives the same energy expression as in Case {\bf a} above.
\vskip.1in
For both case {\bf a} and case {\bf b} above, the structure function corresponding to the $(p+1)$-dimensional unirreps of the deformed oscillator algebra is simply $\Phi^{(III}_{a,b}(z)=z(z-p-1)$.

\subsubsection{Potential $ V_2(u,v)$}

The integrals of motion of the  superintegrable system associated to the potential $V_2$ with Hamiltonian $\hat{\mathcal{H}} = \frac{\exp(2u)}{4(\exp(u)) + 1}  \left( \partial_u^2 + \partial_v^2 \right) + V_2(u,v)$ in Darboux space III are  given by $\cite{MR2023556}$, \begin{align*}
     & A = u^2 \partial_v^2 - 2 u v \partial_u \partial_v + v^2 \partial_u^2 + \frac{b_1 v^2}{4 u^2} + \frac{b_2 u^2}{4 v^2} , \\
     & B = \frac{(2 + v^2) \partial_u^2 -(2 + u^2)\partial_v^2}{2(4 + u^2 + v^2)} + \frac{2b_1 v^2 (v^2 + 2) - 2b_2 u^2(u^2 +2) + b_3(v^2 -u^2)}{4(4 + u^2+ v^2)}.
 \end{align*} These integrals form the quadratic algebra of the form \begin{align*}
    & [A,B] = C, \\
    & [A,C] = - 2 \{A,B\} - (b_1 + b_2 +1) B +  (b_1 - b_2) \hat{\mathcal{H}} +  \frac{(b_2 - b_1)b_3}{4}, \\
    & [B,C] = -2 B^2 - (b_1 + b_2 + 1) B + (b_1 -b_2) \hat{\mathcal{H}} + \frac{(b_2 - b_1)b_3}{4}.
\end{align*} By a direct calculation, we find the Casimir operator of the algebra \begin{align*}
    K_2 =&~ C^2 + 2\{A,B^2\} + (b_1 + b_2 + 5) B^2 - 4 \hat{\mathcal{H}}A^2 - 2(b_1- b_2)\hat{\mathcal{H}} B \\
    & - b_3(b_2 - b_1) B - 4 \hat{\mathcal{H}}A + (2b_3 -1) \hat{\mathcal{H}} A - \frac{b_3^2}{4} A.
\end{align*}  With the differential realization of $A, B$ and in terms of $\hat{\mathcal{H}}$, the Casimir $K_2$  takes the simple form   
\begin{align*}
    K_2   = -(b_1 + b_2 - 2)\hat{\mathcal{H}}^2 + \left(\frac{(b_3 + \frac{3}{2})(b_1 + b_2)}{2} - b_3 -b_1b_2 - \frac{1}{2} \right) \hat{\mathcal{H}} - \frac{b_3^2(b_1 + b_2 -2)}{16}.
 \end{align*}  

The quadratic algebra can be transformed into the deformed oscillator algebra via the realization (i.e., change of basis)
\begin{align*}
 &     A  = - \left((\mathcal{N} + \eta)^2 - \frac{1}{4} + \frac{b_1 + b_2 + 1}{4} \right),  \\
 & B =    - \frac{    (b_1 - b_2) \hat{\mathcal{H}} +  \frac{(b_2 - b_1)b_3}{4}}{ 16\left((\mathcal{N} + \eta)^2 - \frac{1}{4}   \right)}  + b^\dagger \rho(\mathcal{N}) + \rho(\mathcal{N}) b, 
\end{align*} 
where 
$$  \rho(\mathcal{N}) = \frac{1}{3\cdot 2^{12} \cdot (-2)^8 (\mathcal{N} + \eta)(1+ \mathcal{N} + \eta )(1 + 2(\mathcal{N} + \eta))^2}.  $$ 
The structure function is given by  \begin{align*}
 \small
     \Phi_2^{(III)}(\mathcal{N},\eta) = &~ 4096 ((2\mathcal{N} + 2\eta-1)^2 \left(b_2^2 (3 b_1+3 b_2+7)-4 \hat{\mathcal{H}} \left(3 b_1^2+b_1 (12 b_2-12 \hat{\mathcal{H}}+11) \right. \right. \\
     & \left. \left. +9 b_2^2-12 b_2 \hat{\mathcal{H}}+25 b_2-28 \hat{\mathcal{H}}+4\right)\right)-48 (1-2 ( \mathcal{N} +  \eta))^2 \left(-\frac{1}{16} b_2^2 (b_1+b_2-2) \right. \\
     & \left.+\frac{1}{2} \hat{\mathcal{H}} \left(\left(b_2+\frac{3}{2}\right) (b_1+b_2)-2 b_1 b_2-2 b_2-1\right)-\hat{\mathcal{H}}^2 (b_1+b_2-2)\right) \\
     & +(2\mathcal{N} + 2\eta-1)^2 \left(12 \mathcal{N}^2+12\mathcal{N} (2 \eta-1)+12 \eta^2-12 \eta-1\right)\\
     &\qquad\qquad \times \left(b_2^2-4 \hat{\mathcal{H}} (2 b_1+4 b_2-4 \hat{\mathcal{H}}+1)\right) \\
     & +12 (b_1-b_2)^2 (b_2-2 \hat{\mathcal{H}})^2-12 \hat{\mathcal{H}} (2 ( \mathcal{N} +  \eta)-3) (2 ( \mathcal{N} +  \eta)+1) (1-2 ( \mathcal{N} +  \eta))^4 ).
 \end{align*} 
 Here $\eta$ is a constant to be determined from the constraints on the structure function.
 
 By acting on Fock basis states $\vert z,E \rangle,$ we can show that the structure function $\Phi_2^{(III)}$ is factorized as  \begin{align*}
   \Phi_2^{(III)}(z,\eta) = &~ \left( z + \eta- \frac{1}{12} \left(6-\sqrt{3} \sqrt{\frac{\delta_1(E) +(b_3-4 E)^2}{E} -8 (b_1+ b_2)+\frac{g(E)}{E \delta_1(E)}+12}\right) \right) \\
   &\left( z + \eta- \frac{1}{12} \left(6+\sqrt{3} \sqrt{\frac{\delta_1(E)+(b_3-4 E)^2}{E} -8 (b_1+ b_2)+\frac{g(E)}{E \delta_1(E)}+12}\right) \right) \\
   & \left( z + \eta-\frac{1}{24} \left(12-\sqrt{6} \sqrt{\frac{f(E)}{E}+\frac{(-1+i\sqrt{3})  \delta_1(E)}{E}-\frac{(1+i\sqrt{3}) g(E)}{E \delta_1(E)}+24}\right)  \right) \\
   & \left( z + \eta-\frac{1}{24} \left(12+\sqrt{6} \sqrt{\frac{f(E)}{E}+\frac{(-1+i\sqrt{3}) \delta_1(E)}{E}-\frac{(1+i\sqrt{3}) g(E)}{E \delta_1(E)}+24}\right)  \right)  \\
     &   \left( z + \eta- \frac{1}{24} \left(12-\sqrt{6} \sqrt{\frac{f(E)}{E}-\frac{(1+i\sqrt{3}) \delta_1(E)}{E}+\frac{(-1+i\sqrt{3})g(E)}{E \delta_1(E)}+24}\right) \right) \\
     & \left( z + \eta-\frac{1}{24} \left(12+\sqrt{6} \sqrt{\frac{f(E)}{E}-\frac{(1+i\sqrt{3}) \delta_1(E)}{E}+\frac{(-1+i\sqrt{3}) g(E)}{E \delta_1(E)}+24}\right)  \right) , 
 \end{align*} 
 where \begin{align*}
    f(E) = &~ 2 \left(b_3^2-8 E b_3-8 (b_1+b_2-2 E) E\right),  \\
    g(E) = &~ b_3^4-16 E b_3^3+4 E (2 b_1+2 b_2+24 E+3) b_3^2-32 E^2 (2 b_1+2 b_2+8 E+3) b_3 ,\\
    & +16 E^2  (b_1^2+14 b_2 b_1+8 (E-3) b_1+b_2^2+16 E^2+8 b_2 (E-3)+12 E+15 ), \\
    \delta_1(E) = &~ \sqrt[3]{\rho_1(E)+\rho_2(E)}, 
 \end{align*}  
 with\begin{align*}
    \rho_1(E) =&~ b_3^6-24 E b_3^5+240 E^2 b_3^4+12 b_1 E b_3^4+12 b_2 E b_3^4+18 E b_3^4-1280 E^3 b_3^3 \\
    &-192 b_1 E^2 b_3^3-192 b_2 E^2 b_3^3-288 E^2 b_3^3+3840 E^4 b_3^2+1152 b_1 E^3 b_3^2+1152 b_2 E^3 b_3^2\\
    &+1728 E^3 b_3^2+48 b_1^2 E^2 b_3^2+48 b_2^2 E^2 b_3^2-288 b_1 E^2 b_3^2 -  480 b_1 b_2 E^2 b_3^2\\
    &-288 b_2 E^2 b_3^2+360 E^2 b_3^2-6144 E^5 b_3-3072 b_1 E^4 b_3-3072 b_2 E^4 b_3-4608 E^4 b_3 \\
    &-384 b_1^2 E^3 b_3-384 b_2^2 E^3 b_3+2304 b_1 E^3 b_3+3840 b_1 b_2 E^3 b_3+2304 b_2 E^3 b_3-2880 E^3 b_3 \\
    & +4096 E^6+3072 b_1 E^5+3072 b_2 E^5+4608 E^5+768 b_1^2 E^4+768 b_2^2 E^4-4608 b_1 E^4\\
    &-7680 b_1 b_2 E^4-4608 b_2 E^4 +5760 E^4+64 b_1^3 E^3+64 b_2^3 E^3+3744 b_1^2 E^3-2112 b_1 b_2^2 E^3\\
    & +3744 b_2^2 E^3-8064 b_1 E^3-2112 b_1^2 b_2 E^3+10944 b_1 b_2 E^3-8064 b_2 E^3+3456 E^3; \\
  \rho_2(E) =&~ \frac{128}{2043}    \left\{\left[- \left(b_3^2-8 E b_3-8 (b_1+b_2-2 E) E\right)^2-12 E  \left((2 b_1+2 b_2+1) b_3^2\right.\right.\right.\\
  & \left.\left.-8 (2 b_1+2 b_2+1) E b_3 -4 E \left(b_1^2-2 (b_2+4 E-4) b_1+b_2^2+8 b_2-8 b_2 E-4 E-5\right) \right) \right]^3\\
  & +262144  \left[b_3^6-24 E b_3^5 
  +6 E (2 b_1+2 b_2+40 E+3) b_3^4-32 E^2 (6 b_1+6 b_2+40 E+9) b_3^3\right.\\
  & +24 E^2 \left(2 b_1^2-4 (5 b_2-12 E+3) b_1+2 b_2^2
   +160 E^2+72 E  +12 b_2 (4 E-1)+15\right) b_3^2 \\
  & -192 E^3 \left(2 b_1^2-4 (5 b_2-4 E+3) b_1+2 b_2^2+32 E^2+24 E
    +4 b_2 (4 E-3)  +15\right) b_3 \\
  & +32 E^3 \left(2 b_1^3+(-66 b_2+24 E+117) b_1^2 -6 \left(11 b_2^2+(40 E-57) b_2-16 E^2+24 E+42\right) b_1 \right.\\
& \left.\left.\left.+2 b_2^3+3 b_2^2 (8 E+39)+12 b_2 \left(8 E^2-12 E-21\right) +4 \left(32 E^3+36 E^2+45 E+27\right)\right) \right]^2 \right\}^\frac{1}{2}. 
 \end{align*}
 
Imposing the constraints $\eqref{eq:constraint}$ which give the $(p+1)$-dimensional unirreps of thedeformed oscillator algebra, we determine the constant $\eta$ and obtain the following algebraic equations for the energies $E$: 
\vskip.1in
{\bf 1.} $\eta_1(E) = \frac{1}{12} \left(6-\sqrt{3} \sqrt{\frac{\delta_1(E) +(b_3-4 E)^2}{E} -8 (b_1+ b_2)+\frac{g(E)}{E \delta_1(E)}+12}\right)$. This $\eta$ value gives five sets of algebraic equations,
\begin{align*}
  & \delta_1(E)+(b_3-4 E)^2+\frac{g(E)}{\delta_1(E)}=\left(12p(p+2)+8(b_1+b_2)\right) E, \\
  & \eta_1(E)-   \frac{1}{24} \left(12+ \epsilon \sqrt{6} \sqrt{\frac{f(E)}{E}+\frac{(-1+i\sqrt{3})  \delta_1(E)}{E}-\frac{(1+i\sqrt{3}) g(E)}{E \delta_1(E)}+24}\right)   + p +1 = 0, \\
  & \eta_1(E)-   \frac{1}{24} \left(12+ \epsilon\sqrt{6} \sqrt{\frac{f(E)}{E}-\frac{(1+i\sqrt{3}) \delta_1(E)}{E}+\frac{(-1+i\sqrt{3})g(E)}{E \delta_1(E)}+24}\right)  + p +1 = 0,
\end{align*}  
where $\epsilon =\pm 1$. Real solutions to each algebraic equation above give the energies of the system.
\vskip.1in
{\bf 2.} $\eta_2(E) = \frac{1}{24} \left(12-\sqrt{6} \sqrt{\frac{f(E)}{E}+\frac{i \left(i+\sqrt{3}\right)  \delta_1(E)}{E}-\frac{i \left(-i+\sqrt{3}\right) g(E)}{E \delta_1}+24}\right) $ and energy spectra from the real solutions of the three sets of algebraic equations
\begin{align*}
  &f(E)+i (i+\sqrt{3})\delta_1(E)-\frac{i \left(-i+\sqrt{3}\right)g(E)}{ \delta_1(E)}= 24p(p+2)E, \\
  & \eta_2(E) - \frac{1}{24} \left( 12 + \epsilon \sqrt{6} \sqrt{\frac{f(E)}{E}-\frac{i \left(-i+\sqrt{3}\right) \delta_1(E)}{E}+\frac{i \left(i+\sqrt{3}\right) g(E)}{E \delta_1(E)}+24} \right) + p+ 1 = 0,
\end{align*} where again $\epsilon=\pm 1$. 
 \vskip.1in
{\bf 3.} $\eta_3(E) =  \frac{1}{24} \left(12-\sqrt{6} \sqrt{\frac{f(E)}{E}-\frac{(1+i\sqrt{3}) \delta_1(E)}{E}+\frac{(-1+i\sqrt{3})g(E)}{E \delta_1(E)}+24}\right)$. This $\eta$ yields the algebraic equation whose real solutions gives other possible energies of the system,
\begin{align*}
   f(E)- (1+i\sqrt{3})\delta_1(E)+\frac{(-1+i\sqrt{3})g(E)}{ \delta_1(E)}= 24p(p+2)E. 
\end{align*}  

It is in general very difficult to solve the above algebraic equations for $E$ analytically  due to their complicated forms.
To demonstrate the existence of real solutions to the above algebraic equations, we have a closer look at cases of restricted model parameter space. As an example, we consider  $b_1=b_2=b_3=h$ for any $h\in\mathbb{R}$. In this case the structure function reduces to  \begin{align*}
    \Phi_2^{(III)}(z,\eta) = & \left( z + \eta - \frac{1}{2} \right)^2\left( z + \eta -  \frac{1}{8} \left(4-\sqrt{\frac{2 h_2(E)-2 h_1(E) }{E}}\right)\right)\\
    &\qquad \left( z + \eta - \frac{1}{8} \left(4+ \sqrt{\frac{2h_2(E)-2h_1(E) }{E}}\right) \right) \\
    &\left( z + \eta - \frac{1}{8} \left(4- \sqrt{\frac{2h_2(E)+2h_1(E) }{E}}\right) \right)\left( z + \eta - \frac{1}{8} \left(4+ \sqrt{\frac{2h_2(E)+2h_1(E)}{E}}\right) \right),
\end{align*} where \begin{align*}
   &  h_1(E) =\sqrt{h^4+16 h^3 E+8 h^2 E (12 E+1)+64 h E^2 (4 E-15)+16 E^2 \left(16 E^2+8 E+17\right)}, \\
   & h_2(E) = h^2- 24 h E+16 E^2+12 E .
\end{align*} 
Imposing the constraints (\ref{eq:constraint}), we determine the constant $\eta$ and the corresponding energies for the model parameters $b_1=b_2=b_3=h$ as follows.
\vskip.1in
{\bf a.} $\eta_1 = \frac{1}{2}$ and 
\begin{align*}
  &  E_{1,\pm} = \frac{4 h^2+12 h p^2+24 h p+15 h+8 p^4+32 p^3+42 p^2+20 p+1\pm 
 m(h,p)}{4 \left(4 h+4 p^2+8 p+3\right)}, 
%    & E_2 = \frac{m(p)+4 h^2+12 h p^2+24 h p+15 h+8 p^4+32 p^3+42 p^2+20 p+1}{4 \left(4 h+4 p^2+8 p+3\right)}, 
\end{align*} where 
$$m(h,p) = \sqrt{\left(4 h^2+3 h \left(4 p^2+8 p+5\right)+8 p^4+32 p^3+42 p^2+20 p+1\right)^2-h^2 \left(4 h+4 p^2+8 p+3\right)^2}.$$ 
It is easy to check that $m(p)$ is real for any $p\in\mathbb{N}^+$ if $h>0$ and so $E_{1,\pm}$ give the energies of the system for the model parameters $b_1=b_2=b_3=h>0$. 

%The corresponded structure functions are \begin{align*}    & \Phi_2(x) = \left(x -\frac{1}{2} \right)^2\left( p+ 1 -x \right)^2   \prod_{j=1}^2\left( x-   \epsilon\frac{\sqrt{2}}{8} \sqrt{\frac{h_2(E)+h_1(E)}{E}} \right) \\    & \Phi_2(x) =\left(x -\frac{1}{2} \right)^2 \left( p+ 1 -x \right)^2 \prod_{j=1}^2\left( x -   \epsilon\frac{\sqrt{2}}{8} \sqrt{\frac{h_2(E)-h_1(E)}{E}} \right). \end{align*}

\vskip.1in
{\bf b.}  $\eta_2(E) = \frac{1}{8} \left(4+ \epsilon \sqrt{\frac{2 h_2(E)-2 h_1(E) }{E}}\right)$ with $\epsilon = \pm 1$. For this $\eta$ value, the energies are 
\begin{align*}
    & E_{2,\pm} = \frac{8 h^2+6 h p^2+12 h p+12 h+p^4+4 p^3+3 p^2-2 p-4\pm n(h,p)}{8 (4 h+p (p+2))}, 
%    & E_4 = \frac{n(p)+8 h^2+6 h p^2+12 h p+12 h+p^4+4 p^3+3 p^2-2 p-4}{8 (4 h+p (p+2))}, 
\end{align*}
where $$ n(h,p) = \sqrt{\left(8 h^2+6 h \left(p^2+2 p+2\right)+p^4+4 p^3+3 p^2-2 p-4\right)^2-4 h^2 (4 h+p (p+2))^2}.$$
It can be checked that $n(p)$ is real for $h>1$ and so $E_{2,\pm}$ give the energies of the system for the model parameters $b_1=b_2=b_3=h>1$.

Other possible energies corresponding to $\eta_2(E)$ are
\begin{align*}
    & E_{3,\pm} = \frac{1}{8} \left(4 \sqrt{(1-4 h) p (p+2)}-2 h+4 p^2+8 p+3\pm l(h,p)\right),
\end{align*} where \begin{align*}
 %   &  k(p,h) = \sqrt{-8p^2(4 z(p)   -13)+4 h \left(4 z(p)-20 p^2-40 p-3\right)+p \left(80-64 z(p)\right)-24 z(p)+16 p^4+64 p^3+9} ; \\
    &l(p,h) = 8 p^2 \sqrt{4 z(h,p)(16p+7)-4 h \left(4 z(h,p)+20 p^2+40 p+3\right)+16 p^4+64 p^3+80p+22},\\
   &z(p,h)=\sqrt{(1-4 h) p (p+2)}.
\end{align*} 
%By a simple calculation, $k(1,-1) = \sqrt{-88 \sqrt{15}-4 \left(4 \sqrt{15}-63\right)-8 \left(4 \sqrt{15}-13\right)+169} \in \mathbb{C}.$ We will omit this case. 
It is easily seen that $l(p)$ and $z(p)$ are real for $h<0$.
Hence, $E_{3,\pm}$ are real for $h < 0$ and give the energies of the system for model parameters $b_1=b_2=b_3=h<0$.

\subsubsection{Potential $ V_3(\mu,\nu)$}

The constants of motion of the superintegrable system with potential $V_3$ and the Hamiltonian $\hat{\mathcal{H}} = \frac{\mu^2 \partial_\mu^2 - \nu^2 \partial_\nu^2}{(\mu + \nu)(2 + \mu -\nu)} + V_3(\mu,\nu)$ are given by $\cite{MR2023556}$, \begin{align*}
        & A = -\frac{4 \mu^2 \nu^2 \left(\partial_\mu+\partial_\nu\right)^2}{  (\mu+\nu)^2}  - c_2 \frac{\mu - \nu}{\mu \nu} - c_3 \frac{(\mu - \nu)^2}{\mu^2 \nu^2}, \\
        & B =\frac{\nu^2 (\mu+2)\mu \partial_\nu^2 -\mu^2(\nu-2)\nu\partial_\mu^2}{(\mu + \nu)(2 + \mu - \nu)} -\frac{4 \mu^2 \nu^2 \left(\partial_\mu+\partial_\nu\right)^2}{  (\mu+\nu)^2}  - \frac{c_1 \mu^2 \nu^2 + c_2 \mu \nu + 2 c_3(1 + \mu - \nu) }{\mu \nu (2 + \mu - \nu)} .
\end{align*} These integrals form the quadratic algebra with the commutation relations  \begin{align*}
    &[A,B] = C, \\
     &[A,C] = -2 \{A,B\} - B  -2 c_1 c_2 + 4 c_2 \hat{\mathcal{H}}, \\
     & [B,C] = 2 B^2 - 8 c_3 \hat{\mathcal{H}}.
\end{align*} The Casimir operator for the quadratic algebra is given by \begin{align*}
        K_3 = C^2 + 2 \{A,B^2\} -16 c_3 \hat{\mathcal{H}} A + 5 B^2 + 4 c_2\left(c_1 - 2  \hat{\mathcal{H}}\right) B,
\end{align*}
which can be expressed in terms of the Hamiltonian as 
$$ K_3 =  16 c_3 \hat{\mathcal{H}}^2 + 4(c_2^2 - 4c_1c_3)\hat{\mathcal{H}} + 4c_1^2 c_3 .$$     

It can be shown that
\begin{align*}
    &A  =   (\mathcal{N} + \eta)^2 - \frac{1}{2}   , \\
    &B=  - \frac{  c_1 c_2 - 2 c_2 \hat{\mathcal{H}}}{ 2\left((\mathcal{N} + \eta)^2 - \frac{1}{4}   \right)} + b^\dagger \rho(\mathcal{N}) + \rho(\mathcal{N}) b,    
\end{align*} where $\eta$ is a constant to be determined and
$$\rho(\mathcal{N}) = \frac{1}{3\cdot 2^{20}  (\mathcal{N} +\eta)(1+ \mathcal{N} + \eta )(1 + 2(\mathcal{N} + \eta))^2},$$ 
convert the quadratic algebra into the deformed oscillator algebra with the structure function 
\begin{align*}
    & \Phi_3^{(III)}(\mathcal{N},\eta) =-786432 \left(-c_1^2+4 c_1 \hat{\mathcal{H}}+\hat{\mathcal{H}} \left(2\mathcal{N}+ 2\eta-1\right)^2-4 \hat{\mathcal{H}}^2\right) \left(c_2^2-c_3 (2\mathcal{N}+ 2\eta-1)^2\right).
 \end{align*} By acting it on a Fock basis $\vert z,E \rangle,$  the structure function becomes  \begin{align*}
    \Phi_3^{(III)} (z ,\eta) = &~ 12582912\, {c_3}\,E  \,\left(z + \eta -\left(\frac{1}{2}-\frac{c_2}{2\sqrt{c_3}}\right)\right)\left(z + \eta -\left(\frac{1}{2}+\frac{c_2}{2\sqrt{c_3}}\right)\right)    \\
    &    \left(z + \eta -\frac{E-\sqrt{E  }(c_1-2 E)}{2 E}\right) \left(z + \eta- \frac{E+\sqrt{E } (c_1-2 E)}{2 E}\right).
\end{align*} 

Imposing the constraint conditions $\eqref{eq:constraint}$, we obtain
\vskip.1in
{\bf 1.} $\eta (E)= \frac{E-\sqrt{E } (c_1-2 E)}{2 E}$ or $\eta (E)= \frac{E+\sqrt{E } (c_1-2 E)}{2 E}$. For both cases, we have
$$ E =\frac{1}{8} \left(4 c_1+(p+1)^2 \pm  (p+1) \sqrt{8 c_1+(p+1)^2}\right), $$which is real for $c_1>0$. This gives the energy spectrum of the system for any model parameters $c_1, c_2, c_3$ with $c_1>0$.  
%with the corresponding structure function \begin{align*}     \Phi_3(x)& = x(p+1 -x) \left(\frac{4}{4 c_1-g_1(p)}-\frac{g_1(p)}{2 \sqrt{8 c_1-2 g_1(p)}}+ \epsilon\frac{c_2}{2 \sqrt{c_3}} -\frac{1}{2} - x \right)  \left(\frac{4}{4 c_1-g_1(p)}-\frac{g_1(p)}{2 \sqrt{8 c_1-2 g_1(p)}} \right. \\     & \left. -\epsilon\frac{c_2}{2 \sqrt{c_3}}-\frac{1}{2} -x \right),\end{align*} where $g_1(p) =(p+1) \left(\sqrt{8 c_1+(p+1)^2}-(p+1)\right)  $ 

\vskip.1in
{\bf 2.}  $\eta_\epsilon =  \frac{1}{2} \left(1+\epsilon\,\frac{c_2}{ \sqrt{c_3}}\right)$,  where $\epsilon = \pm 1.$ The corresponding energies are given by 
\begin{align*}
    E_\epsilon  =   \frac{1}{8} \left[4c_1+\left(2(p+1)+\epsilon\frac{c_2}{\sqrt{c_3}}\right)^2\pm \left(2(p+1)+\epsilon\frac{c_2}{\sqrt{c_3}}\right)\sqrt{8c_1+\left(2(p+1)+\epsilon\frac{c_2}{\sqrt{c_3}}\right)^2}\right],
\end{align*} which is real for $c_1>0$ and $c_3>0$.  

% The corresponded structure function is then given by $$\Phi_3(x) = x (x - 2\epsilon \sqrt{c_1})(p+1 -x)\left( 1 + \epsilon \frac{c_2}{\sqrt{c_3}} + p -x\right).$$

\subsubsection{Potential $V_4(\mu,\nu)$}

 For a superintegrable system with the Hamiltonian 
 $\hat{\mathcal{H}} = \frac{\mu^2 \partial_\mu^2 - \nu^2 \partial_\nu^2}{(2 + \mu -\nu)(\mu + \nu)} + V_4(\mu,\nu)$ associated to the potential $V_4$,  the constants of motion are given by $\cite{MR2023556}$ \begin{align*}
       A &= \frac{\nu^2 (\mu+2)\mu \partial_\nu^2 -\mu^2(\nu-2)\nu\partial_\mu^2}{(\mu + \nu)(2 + \mu - \nu)} - \frac{\mu \nu \left(d_1(\nu-2) + d_2(\mu+2) + 2d_3(\nu-\mu + \mu \nu) \right)}{(\mu + \nu)(2 + \mu - \nu)}, \\
      B &= \frac{1}{4 \mu  \nu  (\mu -\nu+2) (\mu+\nu)^2}\left[\left( \mu^ 4 (12 \nu^3-12 \nu^2+\nu +1)+2 \mu^3  \nu -(\nu -1) \mu^2\nu ^2\right) \partial_\mu^2  \right.  \\
     &  +\mu  \nu  (\mu -\nu+2) \left(  \mu^ 2 (12 \nu ^2+1)+2 \mu  \nu +\nu ^2\right) \partial_\mu  \partial_\nu \\
     & \left.+\nu ^2 \left(  \mu^3 (12 \nu ^2-1)+  \mu^2 (12 \nu ^2-1)+\mu (\nu -2) \nu -\nu ^2\right) \partial_\nu^2\right] \\
     & - \frac{(\mu-\nu) \left[(\mu-\nu)(d_1 \mu + d_2 \nu) - 2d_3(\mu^2 + \nu^2 + \mu\nu(2 + \mu -\nu)) \right]}{4 \mu \nu (\mu+ \nu)(2 + \mu -\nu)}.
 \end{align*} They satisfy the quadratic algebra relations
\begin{align*}
     & [A,B] = C , \\
     &[B,C] = - 2 B^2 + 2 \hat{\mathcal{H}} B - \frac{d_3^2}{2}, \\
     & [A,C] = 2 \{A,B\} - 2 \hat{\mathcal{H}} A -B + (d_1 + d_2 + \frac{1}{2} ) \hat{\mathcal{H}} - \frac{d_1 d_2}{2} - 2 \hat{\mathcal{H}}^2.
 \end{align*} 
 By a direct calculation, we find the Casimir operator of this algebra
\begin{align*}
     K_4 = C^2 - 2 \{A,B^2\} + 5 B^2 + 2 \hat{\mathcal{H}}\{A,B\} - d_3^2 A + \left(4 \hat{\mathcal{H}}  - (2 d_1 + 2 d_2 + 5) \hat{\mathcal{H}} +d_1 d_2 \right) B.
 \end{align*} By means of the differential operator representation of $A$ and $B$ above, the Casimir operator $K_4$ can be expressed in terms of $\hat{\mathcal{H}}$ as 
\begin{align*}
     K_4 = 4 \hat{\mathcal{H}}^3 - (2 d_1 + 2 d_2 + 1) \hat{\mathcal{H}}^2 + \left(\frac{(d_1 + d_2)^2}{4} + d_3(d_2 - d_1)\right) \hat{\mathcal{H}} - \frac{d_3 (d_3 - d_1^2 + d_2^2)}{4}.
 \end{align*}   
 We can show that 
\begin{align*}
   A  = &  (\mathcal{N} + \eta)^2 , \\
   B =   &\frac{ \hat{\mathcal{H}} }{2  } - \frac{   -4 \hat{\mathcal{H}}   + 8 (d_1 + d_2 + \frac{1}{2} ) \hat{\mathcal{H}} - 4 d_1 d_2 }{32\left((\mathcal{N} + \eta)^2 - \frac{1}{4}   \right)} + \rho(\mathcal{N})b^\dagger + b \rho(\mathcal{N})  
 \end{align*} with $$\rho(\mathcal{N}) = \frac{1}{3\cdot 2^{20}   (\mathcal{N} + \eta)(1+ \mathcal{N} + \eta )(1 + 2(\mathcal{N} + \eta))^2} , $$
give a realization of the quadratic algebra in terms of the deformed oscillator algebra,   with the structure function given by
  \begin{align*}
    \Phi_4^{(III)} (\mathcal{N},\eta)  = & 16384 (1-2 (\mathcal{N}+\eta))^2  \left(3 \left(d_3 \left(-d_1^2+d_2^2+d_3\right)  +4 d_3  \hat{\mathcal{H}}  (d_1-d_2) \right.\right.\\
    & \left. +4  \hat{\mathcal{H}}^2 (2 d_1+2 d_2+1)- \hat{\mathcal{H}}  (d_1+d_2)^2-16  \hat{\mathcal{H}}^3\right)+6 d_1  \hat{\mathcal{H}}  (d_2-2  \hat{\mathcal{H}} ) \\
    & \left.-4  \hat{\mathcal{H}}^2 (3 d_2-6  \hat{\mathcal{H}} +2)+\left( \hat{\mathcal{H}}^2-d_3^2\right) \left(12 (\mathcal{N}+\eta)^2-12 (\mathcal{N}+\eta)-1\right)-7 d_3^2 \right).
 \end{align*}  
 Here $\eta$ is a constant parameter to be determined. Acting on the Fock states $|z,E\rangle$, the structure function is factorized as
 \begin{align*}
     \Phi_3^{(III)} (z,\eta) = \left(z + \eta - \frac{1}{2} \right)^2 \left(z + \eta -\left(\frac{1}{2}-\frac{\gamma_1(E)}{2 \left(d_3^2-E^2\right)}\right)\right)\left(z + \eta - \left(\frac{1}{2}+\frac{\gamma_1(E)}{2 \left(d_3^2-E^2\right)}\right)\right),
 \end{align*} where $$\gamma_1(E) = \sqrt{d_3^2-E^2}\times\sqrt{-d_1^2 (d_3+E)+4 d_1 E (d_3+E)+d_2^2 (d_3-E)+4 d_2 E (E-d_3)-8 E^3} .$$ 
 From the constraints $\eqref{eq:constraint}$, we determine the constant $\eta$ and the energy spectrum $E$ of the system. We list the results as follows.
 \vskip.1in
 {\bf Case 1}:  $ \eta(E) = \frac{1}{2}-\frac{\gamma_1(E)}{2 \left(d_3^2-E^2\right)}$. Corresponding to this $\eta$ value, we have either
\begin{align}
&\sqrt{ -d_1^2 (d_3+E)+4 d_1 E (d_3+E)+d_2^2 (d_3-E)+4 d_2 E (E-d_3)-8 E^3 }= (p+1)\sqrt{\left(d_3^2-E^2\right)}\label{eq:2.3.4-case1a}
\end{align}
or
\begin{align}
&\sqrt{ -d_1^2 (d_3+E)+4 d_1 E (d_3+E)+d_2^2 (d_3-E)+4 d_2 E (E-d_3)-8 E^3 }= 2(p+1)\sqrt{\left(d_3^2-E^2\right)}. \label{eq:2.3.4-case1b}
%   \gamma_1(E) =(p + 1) \left(d_3^2-E^2\right),\quad {\rm or}\quad
%   \gamma_1(E) =2 (p + 1) \left(d_3^2-E^2\right).
 \end{align} 
 Notice that obviously the solution space of the algebraic equation (\ref{eq:2.3.4-case1b}) is subspace of that of (\ref{eq:2.3.4-case1a}) and so the energy spectrum of the system corresponding to $\eta_1(E)$ is given by solutions to (\ref{eq:2.3.4-case1a}). 
 
 \vskip.1in
 {\bf Case 2:} $\eta = \frac{1}{2}$. In this case, $E$ satisfies the same algebraic equation as (\ref{eq:2.3.4-case1b}) and so do not give new energies of the system. 
    
\vskip.1in 
 {\bf Case 3:} $\eta(E)   = \left(\frac{1}{2}+\frac{\gamma_1(E)}{2 \left(d_3^2-E^2\right)}\right) $. This $\eta$ value give the same equations for $E$ as those in Case 1 above.
 
Due to the complexity of the algebraic equations, it is hard to see whether or not they lead to real energies $E$ for general model parameters. However, we can show that when the model parameter $d_3=0$, the structure function reduces to
\begin{align*}
       \Phi_3^{(III)} (z,\eta) & = \left(z + \eta - \frac{1}{2} \right)^2 \left(z + \eta - \frac{1}{2 E}\left(E-\sqrt{E \left(d_1^2-4 d_1 E+d_2^2-4 d_2 E+8 E^2\right)}\right)\right) \\
       & \left(z + \eta - \frac{1}{2 E}\left(E+\sqrt{E \left(d_1^2-4 d_1 E+d_2^2-4 d_2 E+8 E^2\right)}\right)\right).
\end{align*} In this case, by imposing the constraints $\eqref{eq:constraint}$ we obtain the parameter $\eta$ and the energies of the system with model parameter $d_3=0$,
\begin{align*}
  &\eta_-(E) = \frac{1}{2E}\left(E- \sqrt{E \left(d_1^2-4 d_1 E+d_2^2-4 d_2 E+8 E^2\right)}\right),\\
  &E_\pm =   \frac{1}{16} \left(4 d_1+4 d_2+(p+1)^2 \pm \sqrt{(p+1)^2\left((p+1)^2+8(d_1+d_2)\right)   -16 \left(d_1-d_2\right)^2}\right). 
\end{align*}
Other $\eta$ values from the constraints give rise to same energies as $E_\pm$ above.

It is clear that both $E_\pm$ are real for $d_1=d_2>0$. So $E_\pm$ give the energy spectrum of the system for model parameters $d_1=d_2>0,\,d_3=0$. The corresponding structure function of the $p+1)$-dimensional unirreps is 
$$\Phi^{(III)_{E_\pm}}(z)=z(z-p-1)\left(z-\frac{1}{2E_\pm}\sqrt{E_\pm \left(d_1^2-4 d_1 E_\pm+d_2^2-4 d_2 E_\pm+8 E_\pm^2\right)}\,\right)^2.$$

\subsection{Darboux Space IV}

In Darboux space IV, there are $3$ different potentials in the separable coordinates $(\mu,\nu),\; (u,v)$ and $(\omega,\varphi)$:
\begin{align*}
   & V_1(\mu,\nu) = -\frac{\sin^2 (2\mu)(4a_1 \exp(2\nu) + 4 a_2 \csc^2 (2\mu) + 4a_3 \exp(4\nu) )}{2 \cos 2\mu + a_4}, \\
   & V_2(u,v)=  -\frac{\sin^2 (2u)( \frac{b_2}{\sinh^2 v} + \frac{b_3}{\cosh^2v} ) + b_1}{2 \cos 2u + b_4} , \\
   & V_3(\omega, \varphi) = \frac{\frac{c_1}{\cos^2 \varphi} + \frac{c_2}{\cosh^2 \omega} + c_3 \left(\frac{1}{\sin^2 \varphi} - \frac{1}{\sinh^2 \omega} \right)}{\frac{c_4+2}{\sinh^2 (2 \omega)} + \frac{c_4-2}{\sin^2 (2\varphi)}},
\end{align*}
where $a_i,\,b_i\, c_i$ are real model parameters.

\subsubsection{Potential $V_1(\mu,\nu)$}

The integrals of motion of the superintegrable system in Darboux space IV with potential $V_1$ and the Hamiltonian $\hat{\mathcal{H}} = - \frac{4 \mu^2 \nu^2 }{(a_4 + 2) \mu^2 + (a_4 -2) \nu^2} +V_1(\mu,\nu)$ are 
\begin{align*}
    A =& \mu^2 \partial_\mu^2 + 2 \mu \nu \partial_\mu \partial_\nu + \nu^2 \partial_\nu^2  + \mu \partial_\mu + \nu \partial_\nu + a_1(\mu^2 + \nu^2) + a_3(\mu^2 + \nu^2)^2 ; \\
     B = &~\frac{4(a_4 + 2) \mu^2 \partial_\mu^2 - 4(a_4 -2) \nu^2 \partial_\nu^2}{(a_4 + 2) \mu^2 + (a_4 -2) \nu^2} \\
     & + \frac{2a_1\left( (a_4 + 2) \mu^2 - (a_4 -2) \nu^2 \right) + 4 a_3 \left( (a_4 + 2) \mu^4 - (a_4 -2) \nu^4 \right) + 16 a_2}{(a_4 + 2)\mu^2 + (a_4 -2)\nu^2}. 
\end{align*} They form the quadratic algebra with the commutation relations given by $\cite{MR2023556}$ \begin{align*}
    & [A,B] = C,\\
    & [A,C] = 8 \{A,B\}_a - 16 B + 32 a_1\hat{\mathcal{H}}, \\
    & [B,C] = - 8 B^2 + 256 a_3 A + 128\, a_3 a_4 \hat{\mathcal{H}} + 32 (a_1^2 + 4 a_3 + 16).
\end{align*} By a direct calculation, we find that the Casimir operator is \begin{align*}
    K_1 = C^2 - 8 \{A,B^2\} + 256 \,a_3 A^2 + 80 B^2 + \left(256\, a_3 a_4 \hat{\mathcal{H}} + 64 (16 a_2 a_3 + a_1^2 + 4 a_3)\right) A - 64\, a_1 \hat{\mathcal{H}} B.
\end{align*} With the differential operator representation of $A$ and $B$, the Casimir operator can be expressed in terms of $\hat{\mathcal{H}}$ as  \begin{align*}
    K_1  = -256\, a_3 \hat{\mathcal{H}}^2 + 64\, a_4(4 a_3 - a_1^2) \hat{\mathcal{H}} + 128 (a_1^2 + 4 a_3 + 8 a_2 a_3 - 2 a_1^2 a_2).
\end{align*}

After a long calculation, we find that the change of basis
\begin{align*}
    A & = 4   (\mathcal{N} + \eta)^2  , \\
    B & =    - \frac{  128 a_1 \hat{\mathcal{H}}}{256 \left((\mathcal{N} + \eta)^2 - \frac{1}{4}   \right)} + \rho(\mathcal{N}) b^\dagger + b \rho(\mathcal{N}),
\end{align*} where $$\rho(\mathcal{N})   = \frac{1}{3\cdot 2^{15}   (\mathcal{N} + \eta)(1+ \mathcal{N} + \eta )(1 + 2(\mathcal{N} + \eta))^2} $$ 
maps the quadratic algebra to the deformed oscillator algebra with the structure function
\begin{align*}  
  \Phi_1^{(IV)}(\mathcal{N},\eta)  = &-805306368\, \left(2 (\mathcal{N}+\eta)-1\right)^2 \\
  &\quad\times  \left(128 (-2 a_1^2 a_2+a_1^2+8 a_2 a_3+4 a_3)+64 \hat{\mathcal{H}} a_4  (4 a_3-a_1^2)-256 a_3 \hat{\mathcal{H}}^2\right)\\
&+131072 \left(12 (\mathcal{N}+\eta)^2-12 (\mathcal{N}+\eta)-1\right) (2 (\mathcal{N}+\eta)-1)^2\\
&\quad\times \left(131072 (a_1^2+16 a_2 a_3+4 a_3)+524288\, a_3 a_4 \hat{\mathcal{H}} +1048576\, a_3\right)\\
&-16384\, \left(2 (\mathcal{N}+\eta)-1\right)^2 \left(-7340032 (a_1^2+16 a_2 a_3+4 a_3)+29360128\, a_3 a_4 \hat{\mathcal{H}}-46137344\, a_3\right) \\
& +51539607552\, a_1^2 \hat{\mathcal{H}}^2+206158430208\, a_3\, (2 (\mathcal{N}+\eta)-3)\, (2 (\mathcal{N}+\eta)+1)\, (2 (\mathcal{N}+\eta)-1)^4.
\end{align*} 

Acting on the Fock basis state $\vert z,E \rangle $, we find that the structure function has the factorization 
%\begin{align*}
%    \Phi_1^{(IV)}(z,\eta) = & \left(z + \eta - \frac{1}{2}+\frac{i a_1}{4 \sqrt{a_3}}\right) \left(z + \eta -\frac{1}{2}-\frac{i a_1}{4 \sqrt{a_3}}\right)  \\
%     & \times \left(z + \eta -\frac{1}{2} \left(1-\frac{1}{2 \sqrt{2}}\sqrt{-\sqrt{(16 a_2+4 E a_4 +4)^2-64 (4 a_2+E (E+a_4 ))}-16 a_2-4 E a_4 +4}\right)\right)  \\
%  & \times \left(z + \eta -\frac{1}{2} \left(1+\frac{1}{2 \sqrt{2}}\sqrt{-\sqrt{(16 a_2+4 E a_4 +4)^2-64 (4 a_2+E (E+a_4 ))}-16 a_2-4 E a_4 +4} \right)\right)   \\
%   & \times \left(z + \eta - \frac{1}{2} \left(1-\frac{1}{2 \sqrt{2}}\sqrt{\sqrt{(16 a_2+4 E a_4 +4)^2-64 (4 a_2+E (E+a_4 ))}-16 a_2-4 E a_4 +4}\right)\right)   \\
%     & \times \left(z + \eta - \frac{1}{2} \left(1+ \frac{1}{2 \sqrt{2}}\sqrt{\sqrt{(16 a_2+4 E a_4 +4)^2-64 (4 a_2+E (E+a_4 ))}-16 a_2-4 E a_4 +4} \right)\right) . 
%\end{align*} 
\begin{align*}
      \Phi_1^{(IV)}(z,\eta) =& \left(z + \eta - \left(\frac{1}{2}-\frac{i a_1}{4 \sqrt{a_3}}\right)\right) \left(z + \eta -\left(\frac{1}{2}+\frac{i a_1}{4 \sqrt{a_3}}\right)\right) \\
    & \left(z + \eta -\frac{1}{2\sqrt{2}} \left(\sqrt{2}-\sqrt{1 -4 a_2 - E a_4-m_1(E)}\right)\right)\\
    & \left(z + \eta -\frac{1}{2\sqrt{2}} \left(\sqrt{2}+\sqrt{1 -4 a_2 -  E a_4-m_1(E)}\right)\right)    \\
   & \left(z + \eta -\frac{1}{2\sqrt{2}} \left(\sqrt{2}-\sqrt{1 -4 a_2 -  E a_4+m_1(E)}\right)\right)\\
   & \left(z + \eta -\frac{1}{2\sqrt{2}} \left(\sqrt{2}+\sqrt{1 -4 a_2 -  E a_4+m_1(E)}\right)\right) ,
\end{align*} 
where
$$m_1(E)=\sqrt{(4 a_2+ E a_4 +1)^2-4 (4 a_2+E (E+a_4 ))}.$$ 
Imposing the constraints $\eqref{eq:constraint}$, we have 
\vskip.1in

{\bf 1.} $\eta_1(E) =\frac{1}{2\sqrt{2}} \left(\sqrt{2} - \sqrt{1 -4 a_2 -  E a_4-m_1(E)} \right)$. This $\eta$ value gives the following two sets of energies and corresponding structure functions
\begin{align*}
   E_{1,\epsilon} = &\frac{1}{2} (p+1)\left( -(p+1)\,a_4+\epsilon \sqrt{ (a_4^2-4)(p+1)^2+16(1-a_2)}\right),\\ 
  \Phi^{(IV)}_{E_{1,\epsilon}} (z) = &z(z-p-1)\left(z - \frac{1}{2\sqrt{2}}\,\sqrt{1 -4 a_2 -  E_{1,\epsilon} a_4-m_1E_{1,\epsilon})}-\frac{i a_1}{ 4 \sqrt{a_3}}\right)\\
  & \left(z - \frac{1}{2\sqrt{2}}\,\sqrt{1 -4 a_2 -  E_{1,\epsilon} a_4-m_1(E_{1,\epsilon})}+\frac{i a_1}{ 4 \sqrt{a_3}}\right)\\
  & \left(z - \frac{1}{2\sqrt{2}}\left(\sqrt{1 -4 a_2 -  E_{1,\epsilon} a_4-m_1(E_{1,\epsilon})} - \sqrt{1 -4 a_2 -  E_{1,\epsilon} a_4 +   m_1(E_{1,\epsilon}}\right)\right)\\
  &\left(z -\frac{1}{2\sqrt{2}}\,\left(\sqrt{1 -4 a_2 +  E_{1,\epsilon} a_4-m_1(E_{1,\epsilon})} + \sqrt{1 -4 a_2 - E_{1,\epsilon} a_4+ m_1(E_{1,\epsilon})}\right) \right),
\end{align*}
\begin{align*}
  E_{2,\epsilon} = &-\frac{1}{a_4+\epsilon\, 4 } \left(4a_2+4 p^2+8 p+3\right), \qquad a_4\neq \pm 4,\\
 \Phi_{E_{2,\epsilon}} (z) = &z(z-p-1) \left(z - \frac{1}{2\sqrt{2}}\,\sqrt{1 -4 a_2 -  E_{2,\epsilon} a_4-m_1E_{2,\epsilon})}-\frac{i a_1}{ 4 \sqrt{a_3}}\right)\\
  & \left(z - \frac{1}{2\sqrt{2}}\,\sqrt{1 -4 a_2 -  E_{2,\epsilon} a_4-m_1(E_{1,\epsilon})}+\frac{i a_1}{ 4 \sqrt{a_3}}\right)\\
   & \left(z - \frac{1}{\sqrt{2}}\left(\sqrt{1 -4 a_2 -  E_{2,\epsilon} a_4-m_1(E_{2,\epsilon})} \right)\right)\\
  &\left(z -\frac{1}{2\sqrt{2}}\,\left(\sqrt{1 -4 a_2 +  E_{2,\epsilon} a_4-m_1(E_{2,\epsilon})} + \epsilon\,\sqrt{1 -4 a_2 - E_{2,\epsilon} a_4+ m_1(E_{2,\epsilon})}\right) \right),
\end{align*} 
where $\epsilon = \pm 1$.  
Notice that the energies $E_{1,\epsilon}$ are real for $a_4^2>4,\,a_2<1$.
\vskip.1in
{\bf 2.} $\eta_2(E) =  \frac{1}{2\sqrt{2}} \left(\sqrt{2}- \sqrt{1 -4 a_2 -  E a_4 +   m_1(E)}\right)$. The energies are the same as those given in case {\bf 1} above

\subsubsection{Potential $ V_2(u,v)$}

The constants of motion of the superintegrable system in Darboux space IV with the Hamiltonian $ \hat{\mathcal{H}} = - \frac{\sin^2(2u) \left( \partial_v^2 + \partial_u^2 \right)}{2 \cos(2 u) + b_4} + V_2(u,v)$ are $\cite{MR2023556}$,  \begin{align*}
   A =&~e^{-2v} \left(\frac{\left(e^{4 v}+1\right) \left(2 b_4 \cos (2 u)+3 \cos (4 u)+1\right)}{2 b_4+4 \cos (2 u)}\partial_v^2-\sin (2 u) \frac{\left(e^{4 v}+1\right) \sin (2 u) }{b_4+2 \cos (2 u)}\partial_u^2 \right) \\
    & +e^{-2v}\left(\sin (2 u)\left(e^{4 v}+1\right) \partial_u +\sin (2 u)\left(e^{4 v}-1\right) \partial_u   \partial_v  +\cos (2 u) \left(e^{4 v}-1\right)  \partial_v \right)\\
    & + \frac{1 }{2 \cos 2u + b_4} \left( 2 b_1 \cosh 2 v + (b_2 + b_3) (4 - b_4^2) + \left(\cos 4u + 2 b_4 \cos 2u +3\right) \left( \frac{b_2}{\sinh^2 v} - \frac{b_3}{\cosh^2v} \right)\right), \\
    B = &~\partial_v^2 + \frac{b_2}{\sinh^2 v} + \frac{b_3}{\cosh^2 v}.
\end{align*} These integrals generates the quadratic algebra the commutation relations as follows \begin{align*}
   [A,B] = &C, \\
   [A,C] = &8 \{A,B\} + 16 b_4(b_2 + b_3) A - 16 B +  32 (b_1 + b_3) \hat{\mathcal{H}} - 16  b_4 (b_2 + b_3), \\
   [B,C] = &8 B^2 + 96 A^2 + \left(  64 b_4 \hat{\mathcal{H}} + (2b_2 - 2 b_3 + b_1 + 3)\right) A + 32 \hat{\mathcal{H}}^2 + 32 b_4 (2b_2 - 2b_3 + 1) \hat{\mathcal{H}} \\
  & + 64 b_1(b_2 - b_3) - 8 (b_4^2 - 4 )(b_2 + b_3)^2 + 32 (b_1 + 2 b_2 - 2 b_3).
\end{align*} By a direct computation, we find the Casimir operator of the algbebra,
\begin{align*}
  K_2 & = C^2 + 64 A^3 - 8 \{A,B^2\} - 16 b_4(b_2 + b_3) \{A,B\} + 64 \left( b_4 \hat{\mathcal{H}} +2b_2 - 2b_3 + b_1 +7 \right) A^2\\
  & + \left(160b_4(b_2 + b_3) - 64 (b_2 + b_3) \hat{\mathcal{H}}\right) B - 64 b_4(2 b_3 -2 b_2 -1) \hat{\mathcal{H}}A \\
  &- 16 \left((b_4^2 -4)(b_2 + b_3)^2 + 8(b_1 + 1)(b_3 - b_2) - 4b_1 + 32 \right)A + 64\hat{\mathcal{H}}^2  A .
\end{align*} With the differential realization of $A$ and $B,$ the Casimir $K_2$ is expressible in terms of $\hat{\mathcal{H}}$ as follows \begin{align*}
    K_2 = &128(b_3 - b_2 + 1) \hat{\mathcal{H}}^2 + 128 b_4 (b_2 - b_3 + 1) \hat{\mathcal{H}} \\
    &+ (128 - 80 b_4^2 - 64 b_1)(b_2 + b_3)^2 - 128(b_1 + 2) (b_3 - b_2 - 1) - 256.
\end{align*} 

After a long computation, we find that the realization
\begin{align*}
  A & = - 4 \left((\mathcal{N} +\eta)^2 - \frac{1}{2} \right), \\
  B & =     \frac{     b_4 (b_2 + b_3) }{8} - \frac{  - 32 \cdot 16   b_4 (b_2 + b_3)   - 256 \cdot (b_1 + 2 b_2 - 2 b_3) }{4\gamma^3\left((\mathcal{N} +\eta)^2 - \frac{1}{4}   \right)} + \rho(\mathcal{N}) b^\dagger + b \rho(\mathcal{N}),  
\end{align*} 
where $$     \rho(\mathcal{N}) = \frac{1}{3\cdot 2^{12} \cdot (-8)^8 (\mathcal{N} + \eta)(1+ \mathcal{N} + \eta )(1 + 2(\mathcal{N} + \eta))^2},$$
converts the quadratic algebra into the deformed oscillator algebra with the structure function
\begin{align*}
  \Phi_2^{(IV)}(\mathcal{N},\eta) =&268435456  \left\{ (2 \mathcal{N}+2\eta-1)^2 \left[48(5 a_2+4 b_1-8) (b_2+b_3)^2\right.\right.\\
  &~ \left.+384\left( (b_1+2) (b_2-b_3+1)+ \hat{\mathcal{H}}^2 (-b_2+b_3+1)+ \hat{\mathcal{H}} b_4  (b_2-b_3+1)-2\right)\right]\\
  &+64\, (2\mathcal{N}+ 2\eta-1)^2 \left(12(\mathcal{N}+\eta)^2-12 (\mathcal{N}+\eta-)-1\right) \\
  &~\times\left(b_1 (2 b_2-2 b_3+3) +b_2^2+2 b_2 (b_3+\hat{\mathcal{H}} b_4 +3)+b_3^2-2 b_3 \hat{\mathcal{H}} b_4 -6 b_3+\hat{\mathcal{H}}^2+3 \hat{\mathcal{H}} b_4 +9\right) \\
 & + (2\mathcal{N}+ 2\eta+1) (2\mathcal{N}+ 2\eta-1)^6 (b_1+2 b_2-2 b_3+\hat{\mathcal{H}} b_4 +6) \\
 &~\times\left[986\, b_1 (b_2-b_3)+448 \left(b_1+2 b_2-2 b_3+ \hat{\mathcal{H}} b_4  (2 b_2-2 b_3+1)\right)\right.\\
 &\quad-112\left(b_4^2-4\right) (b_2+b_3)^2+448 \hat{\mathcal{H}}^2+ 96\, b_4  (b_2+b_3)\, \left(2 \hat{\mathcal{H}} (b_1+b_3)- b_4  (b_2+b_3)\right)\\
 &\quad\left. +704 \left(b_1+2 b_2-2 b_3+\hat{\mathcal{H}} b_4 +3\right)-32 b_4^2 (b_2+b_3)^2+192\right] \\
  & \left. +192 \left(2\mathcal{N}+ 2\eta-3 + \hat{\mathcal{H}}^2 (b_1+b_3)^2- (2 \mathcal{N}+2\eta-1)^4 \left(4 (\mathcal{N}+\eta)^2-4 (\mathcal{N}+\eta)-3\right)^2 \right)\right\}.
\end{align*}  
Acting on Fock basis states $\vert z, E \rangle,$ the structure function is factorized as follows
\begin{align*}
    \Phi_2^{(IV)} (z,\eta)  =& \left(z +\eta -\frac{1}{2\sqrt{2}}  \left(\sqrt{2}-\sqrt{1-2b_2+2b_3- \sqrt{(1-4b_2)(1+4b_3)}}\right)\right) \\
    &\left(z +\eta -\frac{1}{2\sqrt{2}} \left(\sqrt{2}+\sqrt{1-2b_2+2b_3- \sqrt{(1-4b_2)(1+4b_3)}}\right)\right) \\
     &\left(z +\eta -\frac{1}{2\sqrt{2}} \left(\sqrt{2}-\sqrt{ 1-2b_2+2b_3+ \sqrt{(1-4b_2)(1+4b_3)}}\right)\right)\\
    &\left(z +\eta -\frac{1}{2\sqrt{2}} \left(\sqrt{2}+\sqrt{ 1-2b_2+2b_3+ \sqrt{(1-4b_2)(1+4b_3)}}\right)\right)  \\
   &  \left(z +\eta -\frac{1}{2 \sqrt{2}} \left(\sqrt{2}-\sqrt{1-b_1-E b_4-m_2(E)}\right)\right)  \\
   & \left(z +\eta -\frac{1}{2 \sqrt{2}} \left(\sqrt{2}+\sqrt{1-b_1-E b_4-m_2(E)} \right)\right)   \\
   &  \left(z +\eta -\frac{1}{2 \sqrt{2}} \left(\sqrt{2}- \sqrt{1-b_1-E b_4+m_2(E)}\right)\right) \\
   &\left(z +\eta -\frac{1}{2 \sqrt{2}} \left(\sqrt{2}+ \sqrt{1-b_1-E b_4+m_2(E)}\right)\right), 
\end{align*}
where $$m_2(E)=\sqrt{(1+b_1+ E b_4 )^2-4 \left(b_1+E^2+E b_4 \right)}.$$
%$\xi = \sqrt{1-2b_2+2b_3- m_2}$ and $\xi' = \sqrt{ 1-2b_2+2b_3+ m_2}$.  

Imposing the constraint conditions $\eqref{eq:constraint}$, we determine the parameter $\eta$ and the corresponding energies of the system. We find
\vskip.1in
{\bf Case 1:} $\eta_{1,\pm}=\frac{1}{2\sqrt{2}}\left(\sqrt{2}-\sqrt{1-2b_2+2b_3\pm \sqrt{(1-4b_2)(1+4b_3)}}\right)$ and the energies $E$ satisfy
\begin{align*}
2\sqrt{2}(p+1)-\,n_{2,\pm}=\sqrt{1-b_1-E b_4+m_2(E)},
\end{align*}
where
$$n_{2,\pm}=\sqrt{1-2b_2+2b_3\pm \sqrt{(1-4b_2)(1+4b_3)}}.$$
Noticing $(1-4b_2)(1+4b_3)=(1-2b_2+2b_3)^2-(2b_2+2b_3)^2\leq (1-2b_2+2b_3)^2$, we conclude that both $n_{2,\pm}$ are real if $b_2< \frac{1}{4},\, b_3>-\frac{1}{4}$.

Solving the algebraic equations give the energies of the system and the corresponding structure functions of the $(p+1)$-dimensional unirreps of the algebra. We have
\begin{align*}
 E_{1\pm, \epsilon}=& -\frac{b_4}{4} \left(2\sqrt{2}(p+1)-\,n_{2,\pm}\right)^2\\
 & +\epsilon \left(2\sqrt{2}(p+1)-\, n_{2,\pm}\right)
 \sqrt{8(1-b_1)+(b_4^2-4)\left(2\sqrt{2}(p+1)-\, n_{2,\pm}\right)^2},\\
\Phi^{(IV)}_{E_{1\pm,\epsilon}}(z,\eta)=& z\,(z-p-1)\,\left(z-\frac{1} {\sqrt{2}}n_\pm\right)\left(z-\frac{1}{2\sqrt{2}}(n_\pm-n_\mp)\right)\left(z-\frac{1}{2\sqrt{2}}(n_\pm+n_\mp)\right)\\
 & \left(z-\frac{1}{2\sqrt{2}}\left(n_\pm-\sqrt{1-b_1-E_{1\pm,\epsilon} b_4-m_2(E_{1\pm,\epsilon})}\right)\right)\\
 &\left(z-\frac{1}{2\sqrt{2}}\left(n_\pm+\sqrt{1-b_1-E_{1\pm,\epsilon} b_4-m_2(E_{1\pm,\epsilon})}\right)\right)\\
 &\left(z-\frac{1}{2\sqrt{2}}\left(n_\pm-\sqrt{1-b_1-E_{1\pm,\epsilon} b_4+m_2(E_{1\pm,\epsilon})}\right)\right),
\end{align*}
where  $\epsilon=\pm 1$.
%Noticing $1-4b_2)(1+4b_3)=(1-2b_2+2b_3)^-2(b_2+b3)^2\leq (1-2b_2+2b_3)^2$, we conclude that both $n_{2,\pm}$ are real if $b_2< \frac{1}{4}\, b_3>\frac{1}{4}$. It follows that 
The energies $E_{1\pm, \epsilon}$ are real for the model parameters $b_2< \frac{1}{4},\, b_3>-\frac{1}{4},\,b_1<1,\, b_4^2>4$. 

\vskip.1in
{\bf Case 2:} $\eta_{1,\pm}=\frac{1}{2\sqrt{2}}\left(\sqrt{2}+\sqrt{1-2b_2+2b_3\pm \sqrt{(1-4b_2)(1+4b_3)}}\right)$ and the energies $E$ satisfy
\begin{align*}
2\sqrt{2}(p+1)+\,n_{2,\pm}=\sqrt{1-b_1-E b_4+m_2(E)}.
\end{align*}
Solutions of the equations give the energies of the system and the corresponding structre functions of the $(p+1)$-dimensional unirreps of the algebra. We have
\begin{align*}
 E_{2\pm, \epsilon}=& -\frac{b_4}{4} \left(2\sqrt{2}(p+1)+\,n_{2,\pm}\right)^2\\
 & +\epsilon \left(2\sqrt{2}(p+1)+\, n_{2,\pm}\right)
 \sqrt{8(1-b_1)+(b_4^2-4)\left(2\sqrt{2}(p+1)+\, n_{2,\pm}\right)^2},\\
\Phi^{(IV)}_{E_{2\pm,\epsilon}}(z,\eta)=& z\,(z-p-1)\,\left(z+\frac{1} {\sqrt{2}}n_\pm\right)\left(z+\frac{1}{2\sqrt{2}}(n_\pm-n_\mp)\right)\left(z+\frac{1}{2\sqrt{2}}(n_\pm+n_\mp)\right)\\
 & \left(z+\frac{1}{2\sqrt{2}}\left(n_\pm-\sqrt{1-b_1-E_{2\pm,\epsilon} b_4-m_2(E_{2\pm,\epsilon})}\right)\right)\\
 &\left(z+\frac{1}{2\sqrt{2}}\left(n_\pm+\sqrt{1-b_1-E_{2\pm,\epsilon} b_4-m_2(E_{2\pm,\epsilon})}\right)\right)\\
 &\left(z+\frac{1}{2\sqrt{2}}\left(n_\pm+\sqrt{1-b_1-E_{2\pm,\epsilon} b_4+m_2(E_{2\pm,\epsilon})}\right)\right),
\end{align*}
where  $\epsilon=\pm 1$. The energies $E_{2\pm, \epsilon}$ are real for the model parameters $b_2< \frac{1}{4},\, b_3>-\frac{1}{4},\,b_1<1,\, b_4^2>4$.

\vskip.1in
{\bf Case 3:} $\eta_{3,-}(E) = \frac{1}{2 \sqrt{2}} \left(\sqrt{2}-\sqrt{1-b_1-E b_4+ m_2(E)}\right)$ and $$\sqrt{2}(p+1)=\sqrt{1-b_1-E b_4+ m_2(E)}$$  or $$2\sqrt{2}(p+1)=\sqrt{1-b_1-E b_4+ m_2(E)}+\sqrt{1-b_1-E b_4- m_2(E)}.$$
The first algebraic equation gives the energies
$$E_{3,1}=\frac{p+1}{2}\left(-(p+1)b_4\pm \sqrt{4(1-b_1)+(b_4^2-4)(p+1)^2}\right),$$
which is real for $b_1<1,\, b_4^2>4$. The corresponding structure function is
\begin{align*}
\Phi^{(IV)}_{E_{3,1}}(z,\eta)=& z\,(z-p-1)\,\left(z-\frac{1}{2\sqrt{2}}\left(\sqrt{1-b_1-E_{3,1} b_4+m_2(E_{3,1})}-n_-\right)\right)\\
&\left(z-\frac{1}{2\sqrt{2}}\left(\sqrt{1-b_1-E_{3,1} b_4+m_2(E_{3,1})}+n_-\right)\right)\\
&\left(z-\frac{1}{2\sqrt{2}}\left(\sqrt{1-b_1-E_{3,1} b_4+m_2(E_{3,1})}-n_+\right)\right)\\
&\left(z-\frac{1}{2\sqrt{2}}\left(\sqrt{1-b_1-E_{3,1} b_4+m_2(E_{3,1})}+n_+\right)\right)\\
&\left(z-\frac{1}{2\sqrt{2}}\left(\sqrt{1-b_1-E_{3,1} b_4+m_2(E_{3,1})}-\sqrt{1-b_1-E_{3,1} b_4-m_2(E_{3,1})}\right)\right)\\
&\left(z-\frac{1}{2\sqrt{2}}\left(\sqrt{1-b_1-E_{3,1} b_4+m_2(E_{3,1})}+\sqrt{1-b_1-E_{3,1} b_4-m_2(E_{3,1})}\right)\right).
\end{align*}
The second algebraic equation yields the energies
$$E_{3,2}=\frac{1}{2-b_4}\left(4(p+1)^2+b_1-1\right),$$
which is well defined for the model parameter $b_4\neq 2$. 
The associated structure function is given by
\begin{align*}
\Phi^{(IV)}_{E_{3,2}}(z,\eta)=& z\,(z-p-1)\,\left(z-\frac{1}{2\sqrt{2}}\left(\sqrt{1-b_1-E_{3,2} b_4+m_2(E_{3,2})}-n_-\right)\right)\\
&\left(z-\frac{1}{2\sqrt{2}}\left(\sqrt{1-b_1-E_{3,2} b_4+m_2(E_{3,2})}+n_-\right)\right)\\
&\left(z-\frac{1}{2\sqrt{2}}\left(\sqrt{1-b_1-E_{3,2} b_4+m_2(E_{3,2})}-n_+\right)\right)\\
&\left(z-\frac{1}{2\sqrt{2}}\left(\sqrt{1-b_1-E_{3,2} b_4+m_2(E_{3,2})}+n_+\right)\right)\\
&\left(z-\frac{1}{2\sqrt{2}}\left(\sqrt{1-b_1-E_{3,2} b_4+m_2(E_{3,2})}-\sqrt{1-b_1-E_{3,2} b_4-m_2(E_{3,2})}\right)\right)\\
&\left(z-\frac{1}{\sqrt{2}}\sqrt{1-b_1-E_{3,2} b_4+m_2(E_{3,2})}\right).
\end{align*}

\vskip.1in
{\bf Case 4:} $\eta_{4,+} =\frac{1}{2 \sqrt{2}} \left(\sqrt{2}+\sqrt{1-b_1-E b_4-m_2(E)}\right)$. We have the algebraic equation  $$2\sqrt{2}(p+1)=\sqrt{1-b_1-E b_4+ m_2(E)}-\sqrt{1-b_1-E b_4- m_2(E)}.$$
Solving, we obtain
$$E_{4}=\frac{1}{2+b_4}\left(4(p+1)^2+b_1-1\right).$$
This gives the energy spectrum of the system for the model parameter $b_4\neq -2$. The corresponding structure function is 
\begin{align*}
\Phi^{(IV)}_{E_4}(z,\eta)=& z\,(z-p-1)\,\left(z+\frac{1}{2\sqrt{2}}\left(\sqrt{1-b_1-E_{4} b_4-m_2(E_4)}-n_-\right)\right)\\
&\left(z+\frac{1}{2\sqrt{2}}\left(\sqrt{1-b_1-E_{4} b_4-m_2(E_{4})}+n_-\right)\right)\\
&\left(z+\frac{1}{2\sqrt{2}}\left(\sqrt{1-b_1-E_{4} b_4-m_2(E_{4})}-n_+\right)\right)\\
&\left(z+\frac{1}{2\sqrt{2}}\left(\sqrt{1-b_1-E_{4} b_4-m_2(E_{4})}+n_+\right)\right)\\
&\left(z+\frac{1}{\sqrt{2}}\sqrt{1-b_1-E_{4} b_4-m_2(E_{4})}\right)\\
&\left(z+\frac{1}{2\sqrt{2}}\left(\sqrt{1-b_1-E_{4} b_4-m_2(E_{4})}+\sqrt{1-b_1-E_{4} b_4+m_2(E_{4})}\right)\right).
\end{align*}

\subsubsection{Potential $V_3(\omega, \varphi)$}

The constants of motion of the superintegrable system in Darboux space IV with the potential $V_3(\omega,\varphi)$ are $\cite{MR2023556}$ \begin{align*}
   A  =& -2c_4 \frac{\partial_{\varphi}^2 + \partial_{\omega}^2}{\frac{c_4 + 2}{\sinh^2(2\omega)} + \frac{c_4 -2}{\sin^2(2\varphi)}}  + \frac{(c_4 + 2) \sin^2(2 \varphi) \partial_{\varphi}^2 - (c_4 - 2) \sinh^2(2\omega) \partial_{\omega}^2}{(c_4 + 2) \sin^2(2 \varphi) + (c_4 - 2) \sinh^2(2\omega) }  \\
   & + \frac{1}{\frac{c_4 + 2}{\sinh^2(2\omega)} + \frac{c_4 - 2}{\sin^2 \omega}} \left[\frac{c_4 + 2}{\sinh^2(2\omega)} \left(\frac{c_3}{\sin^2\varphi} + \frac{c_1}{\cos^2\varphi} \right)+\frac{c_4 - 2}{\sin^2(2\omega)}\left(\frac{c_3}{\sinh^2\omega} - \frac{c_2}{\cosh^2\omega}\right) \right], \\
    B = &\frac{1}{2}\,  \sin (2\varphi) \sinh (2\omega ) \tan (\varphi-  i\omega ) \tan (\varphi +i \omega )\, \left(\cot (2\varphi) \,\partial_{\omega}^2 +\coth (2\omega )\, \partial_{\varphi}^2\right) \\
   & +\left[-i \cos (2\varphi) \sinh (2\omega ) \sinh \left(\log \left(\frac{\tan (\varphi-  i\omega )}{\tan (\varphi +i \omega )}\right)\right)\partial_{\omega} \right. \\
   & \quad \left.+i \cosh (2\omega ) \sin (2\varphi)  \sinh \left(\log\left(\frac{\tan (\varphi-  i\omega )}{\tan (\varphi + i \omega)}\right)\right)\partial_{\varphi} +2 \cosh \left(\log \left(\frac{\tan (\varphi-  i\omega )}{\tan (\varphi + i\omega)}\right)\right)\right] \\
   &+ \frac{1}{\frac{c_4 + 2}{\sinh^2 2\omega} + \frac{c_4 -2 }{\sin^2 \omega}} \left[ \frac{c_4 + 2}{\sinh^2 2 \omega}\left( c_1 \cosh 2\omega \tan^2 \varphi - c_2 \cos 2\varphi - \frac{c_3 \left(2 \cos^2 \varphi(\sinh^2 \omega - \sin^2 \varphi) \right)+1}{\sin^2 \varphi}\right) \right. \\
   &\quad \left. + \frac{c_4 -2}{\sin^2 2 \varphi} \left( c_2 \cos 2 \varphi \tanh^2 \omega + c_1 \cosh 2 \omega - \frac{c_3 \left( 2 \cosh^2 \omega(\sinh^2 \omega - \sin^2 \varphi) + 1\right)}{\sinh^2 \omega} \right) \right].
\end{align*}    These integrals form the quadratic algebra with the following commutation relations \begin{align*}
   [A,B] =& C,\\
   [A,C] =&  -8 \{A,B\} - 16 B    - 16 (c_1 -c_3)(c_2 - c_3), \\
   [B,C] =& -24 A^2 + 8 B^2 + 16\left(2c_4 \hat{\mathcal{H}} - 2c_1 +2c_2 +3\right)A-16\left[(c_4+ 2)c_1 + (c_4-2)c_2 - c_4 + 64 c_3\right] \hat{\mathcal{H}}  \\
  & - 8 (c_4^2 -4)\hat{\mathcal{H}}^2 - 8 c_1^2 - 8 c_2^2+ 16 c_3^2+32 c_1 c_2 + 48 c_3 (c_1 + c_2)  - 16 (c_1 - c_2),
\end{align*} which is the symmetry algebra of the superintegrable system. We can calculate the Casimir operator of the algebra \begin{align*}
   K_3 = &C^2 - 16 A^3  + 8\{A,B^2\} -16 (2c_2 - 2c_1 - 7) A^2 + 80 B^2 \\
   &- 16\left(c_4^2 - 4  + 2(c_4 +2) c_1 -2(c_4 -2)c_2 + 8 c_3 + 2c_4 \right)\hat{\mathcal{H}} \\
   &+ 16 \left( c_1^2 + c_2^2 - 2 c_3^2 - 6c_3(c_1 + c_2) - 4 c_1c_2 + 2 c_1 - 2c_2 - 8 \right) A   + 32 (c_2 - c_3)(c_1 - c_3)B.
\end{align*}  We can show that in terms of the Hamiltonian the Casimir $K_3$ takes the simple form \begin{align*}
   K_3 = &16(c_4^2 - 4) \hat{\mathcal{H}}^2
 - 16 \left(
(c_4 + 2)((c_1 - c_3)^2 - 2c_1) + (c_4-2)((c_2-c_3)^2 + 2c_2)- 8c_3- 4 c_4 \right) \hat{\mathcal{H}} \\
& - 32(c_1 - c_2)(3c_3^2 - c_1c_2 - c_3(c_1 + c_2))
+ 32(c_1^2 + c_2^2- 4c_3(c_1 + c_2) - 2c_1c_2 + 2c_1 - 2c_2). 
\end{align*}  

After a long computation, we find that the realization \begin{align*}
  A(\mathcal{N}) =-4& (\mathcal{N} + \eta)^2 , \quad  B =  - \frac{        (c_1 -c_3)(c_2 - c_3)}{ 4\left((\mathcal{N} + \eta)^2 - \frac{1}{4}   \right)} + \rho(\mathcal{N}) b + b^\dagger \rho(\mathcal{N})  
\end{align*}
with
$$ \rho(\mathcal{N}) = \frac{1}{3\cdot 2^{12} \cdot (-8)^8 (\mathcal{N} + \eta)(1+ \mathcal{N} + \eta )(1 + 2(\mathcal{N} + \eta))^2},$$
changes the quadratic algebra to the deformed oscillator algebra with the  structure function
\small
\begin{align*}
\Phi_3^{(IV)}(\mathcal{N},\eta) =& 268435456  \left\{\, 16\, \left(12 \mathcal{N}^2+12 \mathcal{N} (2 \eta-1)+12 \eta^2-12 \eta-1\right) (2\mathcal{N}+ 2\eta-1)^2 \right.\\
& \quad\times \left[c_1^2+c_1 \left(-4 c_2-6 c_3+2 \hat{\mathcal{H}} c_4 +4 \hat{\mathcal{H}}-2\right) +c_2^2+c_2 \left(-6 c_3-2 \hat{\mathcal{H}} c_4 +4 \hat{\mathcal{H}}+2\right)\right.\\
&\qquad \quad\left.-2 c_3^2+128 c_3 \hat{\mathcal{H}}+\hat{\mathcal{H}}^2 c_4^2-4 \hat{\mathcal{H}}^2+6 \hat{\mathcal{H}} c_4 +9\right] \\
&+16\, (2\mathcal{N}+ 2\eta-1)^2 \left[7 c_1^2-2 c_1 \left(14 c_2+21 c_3-7 \hat{\mathcal{H}} c_4 -14 \hat{\mathcal{H}}+4\right)+7 c_2^2\right.\\
&\quad \left.+c_2 \left(-42 c_3-14 \hat{\mathcal{H}} (c_4 -2)+8\right)
-14 c_3^2+896 c_3 \hat{\mathcal{H}}+7 \hat{\mathcal{H}}^2 c_4^2-28 \hat{\mathcal{H}}^2+36 \hat{\mathcal{H}} c_4 +36\right]\\
&-3\, (2  \mathcal{N}+ 2 \eta-1)^2 \left[352\, \hat{\mathcal{H}} \left(-c_1+(c_4 -2) (c_2-c_3)^2+2 c_2(c_4-2)-c_3^2-8 c_3-4 c_4 \right)\right.\\
&\left.+32 (c_1-c_2) \left(c_1 (c_2+c_3)+c_3 (c_2-3 c_3)\right)-16 \hat{\mathcal{H}}^2 \left(c_4^2-4\right)\right]+48 (c_1-c_3)^2 (c_2-c_3)^2 \\
& -96 \,(2\mathcal{N}+ 2\eta-3) (2\mathcal{N}+ 2\eta+1) (2\mathcal{N}+ 2\eta-1)^4 (c_1-c_2-\hat{\mathcal{H}} c_4 -3)\\
&\left.+48\, (2 \mathcal{N}+ 2 \eta-1)^4 \left[( 2\mathcal{N}+ 2 \eta-1)^4-8 ( 2\mathcal{N}+ 2 \eta-1)^2+16\right] \right\}.
\end{align*} 

The structure function is a polynomial of degree 8 in $\mathcal{N}$. Acting on the Fock basis states $\vert z, E \rangle$, it becomes a polynomial of degree 8 in $z$. In order to determine the energy spectrum of the superintegrable system, we have to find the finite-dimensional unirreps of the deformed oscillator algebra by solving the constraints. This requires the factorization of the structure function. However, it turns out to be very difficult to factorize the structure function for general model parameters $c_i$ (even using symbolic computation softwares such as  Mathematica). 
In the following we restrict our attention to special model parameters and present analytic and closed-form results for $c_1 = c_2 =1,\,c_3 = c_4 =0.$ 
%\vskip.2in
%\noindent{\bf Case 2:} $c_1 = c_2 =1,\,c_3 = c_4 =0.$ 
%\vskip.1in

In this case, we find that the structure function factorizes as 
\begin{align*}
    \Phi_3^{(IV)}(z,\eta) = &\left( z + \eta - \frac{1}{2\sqrt{2}} \left(\sqrt{2}- \sqrt{-\sqrt{4 E^2-12 E+5}-2 E+3}\right)\right) \\
    &\left( z + \eta - \frac{1}{2\sqrt{2}} \left(\sqrt{2}+ \sqrt{-\sqrt{4 E^2-12 E+5}-2 E+3}\right)\right) \\
   & \left( z + \eta -\frac{1}{2\sqrt{2}} \left(\sqrt{2}- \sqrt{\sqrt{4 E^2-12 E+5}-2 E+3}\right) \right) \\
   &\left( z + \eta -\frac{1}{2\sqrt{2}} \left(\sqrt{2} +\sqrt{\sqrt{4 E^2-12 E+5}-2 E+3}\right) \right)\\
    &  \left( z + \eta -\frac{1}{2\sqrt{2}} \left(\sqrt{2}- \sqrt{-\sqrt{4 E^2-4 E-3}+2 E-1}\right) \right) \\
    &\left( z + \eta - \frac{1}{2\sqrt{2}} \left(\sqrt{2}+ \sqrt{-\sqrt{4 E^2-4 E-3}+2 E-1}\right)\right) \\
    &\left( z + \eta -\frac{1}{2\sqrt{2}} \left(\sqrt{2}- \sqrt{\sqrt{4 E^2-4 E-3}+2 E-1}\right) \right) \\
    &\left( z + \eta -\frac{1}{2\sqrt{2}} \left(\sqrt{2}+ \sqrt{\sqrt{4 E^2-4 E-3}+2 E-1}\right) \right).
\end{align*} 
Imposing the constraints $\eqref{eq:constraint}$, we determine the constant $\eta$ and obtain the energies of the system and the structure function of the symmetry algebra for the model parameters $c_1=c_2=1,\, c_3=c_4=0$ as follows.
\vskip.1in
{\bf a.} The constant $\eta$ is $\eta_{a}(E)=\frac{1}{2\sqrt{2}} \left(\sqrt{2}- \sqrt{-\sqrt{4 E^2-12 E+5}-2 E+3}\right)$ and the energies are given by the equation
\begin{align*}
 2\sqrt{2}(p+1)&= \sqrt{\sqrt{4 E^2-12 E+5}-2 E+3}+  \sqrt{-\sqrt{4 E^2-12 E+5}-2 E+3} \\
 &\Longrightarrow \quad E_{a}=-2(p+1)^2+\frac{5}{2}.
 \end{align*}
The associated structure function is 
\begin{align*}
    \Phi_{E_{a}}^{(IV)}(z,\eta) = &\,z\,(z-p-1)\left( z - \frac{1}{\sqrt{2}} \sqrt{-\sqrt{4 E_{a}^2-12 E_{2a}+5}-2 E_{a}+3}\right) \\
   & \left( z -\frac{1}{2\sqrt{2}} \left( \sqrt{-\sqrt{4 E_{a}^2-12 E_{a}+5}-2 E_{a}+3}- \sqrt{\sqrt{4 E_{a}^2-12 E_{a}+5}-2 E_{a}+3}\right) \right) \\
    &  \left( z -\frac{1}{2\sqrt{2}} \left(\sqrt{-\sqrt{4 E_{a}^2-12 E_{a}+5}-2 E_{a}+3} - \sqrt{-\sqrt{4 E_{a}^2-4 E_{a}-3}+2 E_{a}-1}\right) \right) \\
   &  \left( z -\frac{1}{2\sqrt{2}} \left(\sqrt{-\sqrt{4 E_{a}^2-12 E_{a}+5}-2 E_{a}+3} + \sqrt{-\sqrt{4 E_{a}^2-4 E_{a}-3}+2 E_{a}-1}\right) \right) \\
   &  \left( z -\frac{1}{2\sqrt{2}} \left(\sqrt{-\sqrt{4 E_{a}^2-12 E_{a}+5}-2 E_{a}+3} - \sqrt{\sqrt{4 E_{a}^2-4 E_{a}-3}+2 E_{a}-1}\right) \right) \\
   &  \left( z -\frac{1}{2\sqrt{2}} \left(\sqrt{-\sqrt{4 E_{a}^2-12 E_{a}+5}-2 E_{a}+3} + \sqrt{\sqrt{4 E_{a}^2-4 E_{a}-3}+2 E_{a}-1}\right) \right).
\end{align*}

 \vskip.1in
{\bf b.} The constant $\eta$ is given by  $\eta_{b}(E)=\frac{1}{2\sqrt{2}} \left(\sqrt{2}- \sqrt{\sqrt{4 E^2-12 E+5}-2 E+3}\right)$ and the energies are 
\begin{align*}
 \sqrt{2}(p+1)=\sqrt{\sqrt{4 E^2-12 E+5}-2 E+3}\qquad\Longrightarrow \quad
 E_{b}=-\frac{1}{2}\left((p+1)^2-3+\frac{1}{(p+1)^2}\right).
\end{align*} The corresponding structure function of the $(p+1)$-dimensional unirreps is 
\begin{align*}
    \Phi_{E_{b}}^{(IV)}(z,\eta) = &\,z\,(z-p-1)\left( z - \frac{1}{\sqrt{2}} \sqrt{\sqrt{4 E_{b}^2-12 E_{b}+5}-2 E_{b}+3}\right) \\
   & \left( z -\frac{1}{2\sqrt{2}} \left( \sqrt{\sqrt{4 E_{b}^2-12 E_{b}+5}-2 E_{b}+3}- \sqrt{-\sqrt{4 E_{b}^2-12 E_{b}+5}-2 E_{b}+3}\right) \right) \\
    &  \left( z -\frac{1}{2\sqrt{2}} \left(\sqrt{\sqrt{4 E_{b}^2-12 E_{b}+5}-2 E_{b}+3} - \sqrt{-\sqrt{4 E_{b}^2-4 E_{b}-3}+2 E_{b}-1}\right) \right) \\
   &  \left( z -\frac{1}{2\sqrt{2}} \left(\sqrt{\sqrt{4 E_{b}^2-12 E_{b}+5}-2 E_{b}+3} + \sqrt{-\sqrt{4 E_{b}^2-4 E_{b}-3}+2 E_{b}-1}\right) \right) \\
   &  \left( z -\frac{1}{2\sqrt{2}} \left(\sqrt{\sqrt{4 E_{b}^2-12 E_{b}+5}-2 E_{b}+3} - \sqrt{\sqrt{4 E_{b}^2-4 E_{b}-3}+2 E_{b}-1}\right) \right) \\
   &  \left( z -\frac{1}{2\sqrt{2}} \left(\sqrt{\sqrt{4 E_{b}^2-12 E_{b}+5}-2 E_{b}+3} + \sqrt{\sqrt{4 E_{b}^2-4 E_{b}-3}+2 E_{b}-1}\right) \right).
\end{align*}

\vskip.1in
{\bf c.} The constant $\eta$ is $\eta_{2c}(E)=\frac{1}{2\sqrt{2}} \left(\sqrt{2}+ \sqrt{-\sqrt{4 E^2-12 E+5}-2 E+3}\right)$ and the energy $E$ satisfies 
\begin{align*}
 2\sqrt{2}(p+1)&= \sqrt{\sqrt{4 E^2-12 E+5}-2 E+3}- \sqrt{-\sqrt{4 E^2-12 E+5}-2 E+3} \\
 &\Longrightarrow \quad E_{c}=-2(p+1)^2+\frac{1}{2}.
 \end{align*} The structure function is given by
\begin{align*}
    \Phi_{E_{c}}^{(IV)}(z,\eta) = &\,z\,(z-p-1)\left( z + \frac{1}{\sqrt{2}} \sqrt{-\sqrt{4 E_{c}^2-12 E_{c}+5}-2 E_{c}+3}\right) \\
   & \left( z +\frac{1}{2\sqrt{2}} \left( \sqrt{-\sqrt{4 E_{c}^2-12 E_{c}+5}-2 E_{c}+3}- \sqrt{\sqrt{4 E_{c}^2-12 E_{c}+5}-2 E_{c}+3}\right) \right) \\
    &  \left( z +\frac{1}{2\sqrt{2}} \left(\sqrt{-\sqrt{4 E_{c}^2-12 E_{c}+5}-2 E_{c}+3} - \sqrt{-\sqrt{4 E_{c}^2-4 E_{c}-3}+2 E_{c}-1}\right) \right) \\
   &  \left( z +\frac{1}{2\sqrt{2}} \left(\sqrt{-\sqrt{4 E_{c}^2-12 E_{c}+5}-2 E_{c}+3} + \sqrt{-\sqrt{4 E_{c}^2-4 E_{c}-3}+2 E_{c}-1}\right) \right) \\
   &  \left( z +\frac{1}{2\sqrt{2}} \left(\sqrt{-\sqrt{4 E_{c}^2-12 E_{c}+5}-2 E_{c}+3} - \sqrt{\sqrt{4 E_{c}^2-4 E_{c}-3}+2 E_{c}-1}\right) \right) \\
   &  \left( z +\frac{1}{2\sqrt{2}} \left(\sqrt{-\sqrt{4 E_{c}^2-12 E_{c}+5}-2 E_{c}+3} + \sqrt{\sqrt{4 E_{c}^2-4 E_{c}-3}+2 E_{c}-1}\right) \right).
\end{align*}

\vskip.1in
{\bf d.}  The constant $\eta$ is $ \eta_{d}(E) = \frac{1}{2\sqrt{2}} \left(\sqrt{2}- \sqrt{-\sqrt{4 E^2-4 E-3}+2 E-1}\right)$ and the energies are
\begin{align*}
 2\sqrt{2}(p+1)&= \sqrt{\sqrt{4 E^2-4 E-3}+2 E-1}+  \sqrt{-\sqrt{4 E^2-4 E-3}+2 E-1} \\
 &\Longrightarrow \quad E_{d}=2(p+1)^2-\frac{1}{2}.
 \end{align*} The corresponding structure function reads 
\begin{align*}
    \Phi_{d}^{(IV)}(z,\eta) = &\,z\,(z-p-1)\\
    &\left( z - \frac{1}{2\sqrt{2}} \left(\sqrt{-\sqrt{4 E_{d}^2-4 E_{d}-3}+2 E_{d}-1}- \sqrt{-\sqrt{4 E_{d}^2-12 E_{d}+5}-2 E_{d}+3}\right)\right) \\
    &\left( z - \frac{1}{2\sqrt{2}} \left(\sqrt{-\sqrt{4 E_{d}^2-4 E_{d}-3}+2 E_{d}-1}+ \sqrt{-\sqrt{4 E_{d}^2-12 E_{d}+5}-2 E_{d}+3}\right)\right) \\
   & \left( z - \frac{1}{2\sqrt{2}} \left(\sqrt{-\sqrt{4 E_{d}^2-4 E_{d}-3}+2 E_{d}-1}- \sqrt{\sqrt{4 E_{d}^2-12 E_{d}+5}-2 E_{d}+3}\right)\right) \\
   &\left( z - \frac{1}{2\sqrt{2}} \left(\sqrt{-\sqrt{4 E_{d}^2-4 E_{d}-3}+2 E_{d}-1}+ \sqrt{\sqrt{4 E_{d}^2-12 E_{d}+5}-2 E_{d}+3}\right)\right) \\
    &  \left( z -\frac{1}{\sqrt{2}} \sqrt{-\sqrt{4 E_{d}^2-4 E_{d}-3}+2 E_{d}-1}\right) \\
    &\left( z -\frac{1}{2\sqrt{2}} \left(\sqrt{-\sqrt{4 E_{d}^2-4 E_{d}-3}+2 E_{d}-1}- \sqrt{\sqrt{4 E_{d}^2-4 E_{d}-3}+2 E_{d}-1}\right) \right).
\end{align*}

\vskip.1in
{\bf e.}  The constant $\eta$ is $ \eta_{e}(E) = \frac{1}{2\sqrt{2}} \left(\sqrt{2}-\sqrt{\sqrt{4 E^2-4 E-3}+2 E-1}\right)$ and the energies and structure functions are given by
\begin{align*}
 &\sqrt{2}(p+1)=\sqrt{\sqrt{4 E^2-4 E-3}+2 E-1}\qquad\Longrightarrow \quad
 E_{e}=\frac{1}{2}\left((p+1)^2+1+\frac{1}{(p+1)^2}\right),\\
%\end{align*}  
%\begin{align*}
    \Phi_{e}^{(IV)}(z,\eta) = &\,z\,(z-p-1)\\
    &\left( z - \frac{1}{2\sqrt{2}} \left(\sqrt{\sqrt{4 E_{e}^2-4 E_{e}-3}+2 E_{e}-1}- \sqrt{-\sqrt{4 E_{e}^2-12 E_{e}+5}-2 E_{e}+3}\right)\right) \\
    &\left( z - \frac{1}{2\sqrt{2}} \left(\sqrt{\sqrt{4 E_{e}^2-4 E_{e}-3}+2 E_{e}-1}+ \sqrt{-\sqrt{4 E_{e}^2-12 E_{e}+5}-2 E_{e}+3}\right)\right) \\
   &\left( z - \frac{1}{2\sqrt{2}} \left(\sqrt{\sqrt{4 E_{e}^2-4 E_{e}-3}+2 E_{e}-1}- \sqrt{\sqrt{4 E_{e}^2-12 E_{e}+5}-2 E_{e}+3}\right)\right) \\
   &\left( z - \frac{1}{2\sqrt{2}} \left(\sqrt{\sqrt{4 E_{e}^2-4 E_{e}-3}+2 E_{e}-1}+ \sqrt{\sqrt{4 E_{e}^2-12 E_{e}+5}-2 E_{e}+3}\right)\right) \\
    &  \left( z -\frac{1}{\sqrt{2}} \sqrt{\sqrt{4 E_{e}^2-4 E_{e}-3}+2 E_{e}-1}\right) \\
    &\left( z -\frac{1}{2\sqrt{2}} \left(\sqrt{\sqrt{4 E_{e}^2-4 E_{e}-3}+2 E_{e}-1}- \sqrt{-\sqrt{4 E_{e}^2-4 E_{e}-3}+2 E_{e}-1}\right) \right).
\end{align*}

\vskip.1in
{\bf f.} The constant $\eta$ is $ \eta_{f}(E) = \frac{1}{2\sqrt{2}} \left(\sqrt{2}+ \sqrt{-\sqrt{4 E^2-4 E-3}+2 E-1}\right)$ and the energies are
\begin{align*}
 2\sqrt{2}(p+1)&= \sqrt{\sqrt{4 E^2-4 E-3}+2 E-1}-  \sqrt{-\sqrt{4 E^2-4 E-3}+2 E-1} \\
 &\Longrightarrow \quad E_{f}=2(p+1)^2+\frac{3}{2}.
 \end{align*} The structure function of the $(p+1)$-dimensional unirreps has the form 
\begin{align*}
    \Phi_{f}^{(IV)}(z,\eta) = &\,z\,(z-p-1)\\
    &\left( z + \frac{1}{2\sqrt{2}} \left(\sqrt{-\sqrt{4 E_{f}^2-4 E_{f}-3}+2 E_{f}-1}- \sqrt{-\sqrt{4 E_{f}^2-12 E_{f}+5}-2 E_{f}+3}\right)\right) \\
    &\left( z + \frac{1}{2\sqrt{2}} \left(\sqrt{-\sqrt{4 E_{f}^2-4 E_{f}-3}+2 E_{f}-1}+ \sqrt{-\sqrt{4 E_{f}^2-12 E_{f}+5}-2 E_{f}+3}\right)\right) \\
   &\left( z + \frac{1}{2\sqrt{2}} \left(\sqrt{-\sqrt{4 E_{f}^2-4 E_{f}-3}+2 E_{f}-1}- \sqrt{\sqrt{4 E_{f}^2-12 E_{f}+5}-2 E_{f}+3}\right)\right) \\
  &\left( z + \frac{1}{2\sqrt{2}} \left(\sqrt{-\sqrt{4 E_{f}^2-4 E_{f}-3}+2 E_{f}-1}+ \sqrt{\sqrt{4 E_{f}^2-12 E_{f}+5}-2 E_{f}+3}\right)\right) \\
    &  \left( z +\frac{1}{\sqrt{2}} \sqrt{-\sqrt{4 E_{f}^2-4 E_{f}-3}+2 E_{f}-1}\right) \\
    &\left( z +\frac{1}{2\sqrt{2}} \left(\sqrt{-\sqrt{4 E_{f}^2-4 E_{f}-3}+2 E_{f}-1}- \sqrt{\sqrt{4 E_{f}^2-4 E_{f}-3}+2 E_{f}-1}\right) \right).
\end{align*}

\section{New superintegrable systems in 2D Darboux spaces}

In this section, we investigate superintegrable systems in 2D Darboux spaces with linear and quadratic or quintic integrals of motion. We will first construct generic cubic and quintic algebras and  derive their Casimir operators and realizations in terms of the deformed oscillator algebras. We  will then present examples of new superintegrable systems in 2D Darboux spaces with cubic symmetry algebras.

\subsection{Generic cubic and quintic algebras generated by linear, quadratic or quintic integrals}
\label{1}

We start with the construction of generic cubic and quintic algebras with  structure coefficients involving the Hamiltonians.
 
Let $\hat{X}_1,\;\hat{Y}_1$ be linear integrals, and let $\hat{X}_2,\;\hat{Y}_2$ be quadratic and cubic integrals, respectively. That is, $\deg \hat{X}_1=1=\deg\hat{Y}_1, \, \deg \hat{X}_2=2,\, \deg \hat{Y}_2=3$.  We define the operators $\hat{F}$ and $\hat{G}$ by $\hat{F}=[\hat{X}_1,\hat{X}_2]$ and $\hat{G}=[\hat{Y}_1,\hat{Y}_2]$. Then $\deg \hat{F} = \deg \hat{X}_1 + \deg \hat{X}_2 -1=2$ and $\deg \hat{G} = \deg \hat{Y}_1 + \deg \hat{Y}_2 -1=3$.
By analysing the degrees of the integrals and applying the Jacobi identity constraint $\cite{MR3205917}$, we obtain the following generic cubic and quintic algebras 
%generated by $\{\hat{X}_1,\,\hat{X}_2\,\hat{F}\}$ and $\{\hat{Y}_1,\,\hat{Y}_2\,\hat{G}\}$, respectively:

\begin{proposition}
\label{3.1}
Integrals $\{\hat{X}_1,\hat{X}_2,\hat{F}\}$ satisfy the cubic commutation relations, 
\begin{align}
  [\hat{X}_1,\hat{X}_2] =&\hat{F}, \nonumber \\
  [\hat{X}_1,\hat{F}]  =& u_1 \hat{X}_1^2 + u_2 \hat{X}_1 + u_3\hat{X}_2 +  u, \nonumber\\
  [\hat{X}_2,\hat{F}]= & v_1 \hat{X}_1^3 +v_2 \hat{X}_1^2 + v_3 \hat{X}_1-u_2 \hat{X}_2 -  u_1\{\hat{X}_1,\hat{X}_2\} + v,   \label{eq:a}
\end{align}  
and integrals $\{\hat{Y}_1,\hat{Y}_2,\hat{G}\}$ form the following quintic commutation relations,
\begin{align}
    [\hat{Y}_1,\hat{Y}_2]   = &\hat{G},\nonumber\\
    [\hat{Y}_1,\hat{K}]   = &\alpha \hat{Y}_1^3 + \beta \hat{Y}_1^2 +  \delta \hat{Y}_1 + \epsilon \hat{Y}_2 + \zeta,\nonumber\\
    [\hat{Y}_2,\hat{K}]   =& a \hat{Y}_1^5 + b \hat{Y}_1^4 + c\hat{Y}_1^3 + d \hat{Y}_1^2 +e \hat{Y}_1   + 
     \frac{1}{2}\left( \alpha\, \epsilon - 2\,\delta\right) \hat{Y}_2 - 
       \frac{3}{2} \alpha \{\hat{Y}_1^2,\hat{Y}_2\} - \beta
       \{\hat{Y}_1,\hat{Y}_2\}  + z, \label{eq:b}
\end{align} where $u_j,v_j,\ldots,\alpha,\ldots,z$ are polynomials of the Hamiltonian $\hat{\mathcal{H}}$. Moreover, the coefficients $v_1$ in (\ref{eq:a})and $a$ in (\ref{eq:b}) are not zero polynomials of $\hat{\mathcal{H}}$.
\end{proposition}
The proof of this proposition is a short and straightforward computation from the Jacobi identity requirement.

%First of all, the commutator relations of cubic algebra has the relation as follows \begin{align*}
%&  [\hat{X}_1,\hat{X}_2] = \hat{F} \\
%& [\hat{X}_1,\hat{F}]  = u_1 \hat{X}_1^2 + u_2 \hat{X}_1 + u_3 \hat{X}_2 +   u_4 \{\hat{X}_1,\hat{X}_2\}_a + u \\
%&  [\hat{X}_2,\hat{F}]=  v_1 \hat{X}_1^3 + v_2 \hat{X}_1^2 + v_3 \hat{X}_1 + v_4 \hat{X}_2 +   v_5 \{\hat{X}_1,\hat{X}_2\}_a + v
%\end{align*}  From the Jacobi relation, we have a constraint $[\hat{X}_1,[\hat{X}_2,\hat{F}]] = [\hat{X}_2,[\hat{X}_1,\hat{F}]]$. It turns out that $v_4 = -u_2 $, $v_5 = -u_1 $ and $  u_4 = 0,$ therefore we obtained $\eqref{eq:a}$. Moreover, we set \begin{align*}
% [\hat{Y}_1,\hat{Y}_2]  & = \hat{K} \\
%    [\hat{Y}_1,\hat{K}] & = \alpha \hat{Y}_1^3 + \beta \hat{Y}_1^2 + \gamma \{\hat{Y}_1,\hat{Y}_2 \}_a + \delta \hat{Y}_1 + \epsilon \hat{Y}_2 + \zeta\\
%    [\hat{Y}_2,\hat{K}] & = a \hat{Y}_1^5 + b \hat{Y}_1^4 + c \hat{Y}_1^3 +  d \hat{Y}_1^2 + e\hat{Y}_1 +  f \hat{Y}_2 +  g\{\hat{Y}_1^2,\hat{Y}_2\}_a + h \{\hat{Y}_1,\hat{Y}_2\}_a  + z .
%\end{align*} Using the Jacobi identity $[\hat{Y}_2,[\hat{Y}_1,\hat{K}]] = [\hat{Y}_1,[\hat{Y}_2,\hat{K}]]$, this implies that  $  \gamma = 0$, $ g = -\frac{3}{2}\alpha, $ $f = \frac{\alpha \epsilon}{2} - \delta$ and $ \beta = - h$. We then obtain $\eqref{eq:b}.$ 

\begin{remark}
 For the polynomials on both sides of the commutation relations $\eqref{eq:a}$ and $\eqref{eq:b}$ to have the same degree, we must have that $v_1,v_2, u_1,u_2, u_3, \alpha, \beta, \epsilon, a, b$ are constants and  
\begin{align*}
 %   u_2  & =  u_2^{(0)} + u_2^{(1)}\hat{\mathcal{H}},  \text{ }
    u  & = u^{(0)} + u^{(1)}\hat{\mathcal{H}}, \qquad \text{ }  v_3   = v_3^{(0)} + v_3^{(1)}\hat{\mathcal{H}}, \qquad \text{ } v = v^{(0)} + v^{(1)}\hat{\mathcal{H}},\\
  \delta & = \delta^{(0)} + \delta^{(1)}\hat{\mathcal{H}}, 
  \qquad \text{ } \zeta = \zeta^{(0)} + \zeta^{(1)}\hat{\mathcal{H}}, \qquad
%\text{ } \beta  = \beta_0 + \beta_1\hat{\mathcal{H}}, \\
%  b & = b_0 + b_1\hat{\mathcal{H}}, \text{ } 
    c = c^{(0)} + c^{(1)}\hat{\mathcal{H}},\\   % + c_2\hat{\mathcal{H}}^2, 
    d & = d^{(0)} + d^{(1)}\hat{\mathcal{H}},  % + d_2\hat{\mathcal{H}}^2,   \\
    \qquad e = e^{(0)} + e^{(1)}\hat{\mathcal{H}} + e^{(2)}\hat{\mathcal{H}}^2, 
%\text{ }   f  = f_0 + f_1\hat{\mathcal{H}} , \text{ } 
  \qquad z = z^{(0)} + z^{(1)}\hat{\mathcal{H}} + z^{(2)}\hat{\mathcal{H}}^2,
 \end{align*} 
 where $u^{(0)}, u^{(1)}, \ldots,$ are constants.
\end{remark}

We now construct the Casimir operators for both polynomial algebras. We have
\begin{proposition}
The Casimir operators $C_{(3)}$ and $C_{(5)}$ for the cubic and quintic algebras are respectively given by \begin{align*}
    C_{(3)}  = &\hat{F}^2 - u_1 \{\hat{X}_1^2,\hat{X}_2\}   -u_2\{\hat{X}_1,\hat{X}_2\} + \frac{v_1}{2} \hat{X}_1^4+\frac{2}{3}v_2 \hat{X}_1^3 + \left(  v_3 + u_1^2  \right) \hat{X}_1^2 + \left(u_1u_2   + 2v \right) \hat{X}_1 - 2u \hat{X}_2  - u_3 \hat{X}_2^2 ,\\
    C_{(5)}  = & \hat{G}^2 - \alpha  \{\hat{Y}_1^3,\hat{Y}_2\} + \beta\{\hat{Y}_1^2,\hat{Y}_2\} -   \delta \{\hat{Y}_1,\hat{Y}_2\}  - \epsilon \hat{Y}_2^2 - 2 \zeta \hat{Y}_2 +\frac{a}{3} \hat{Y}_1^6 + \frac{2}{5}b \hat{Y}_1^5 \\ 
    & + \frac{1}{2} \left( c + \frac{5}{3}a \epsilon + 3 \alpha \delta \right) \hat{Y}_1^4  +  \left(2\beta(\delta + 3\alpha) + \frac{2}{5}\epsilon b -2d\right)  \hat{Y}_1^3 \\
    & +  \left[\frac{1}{6}\left(5 - 6a \right)\epsilon^2 + e + \frac{1}{2}\epsilon c + \beta^2 - \frac{3}{4}\alpha\left(\alpha \epsilon -2 \delta  \right) \right] \hat{Y}_1^2 + \left( 2z + \beta \delta + \beta\epsilon(\alpha + \delta) - \frac{1}{5}b\epsilon^2 - \epsilon d \right)  \hat{Y}_1  .
\end{align*}
\end{proposition}

\begin{proof}
By analysinf the degrees of the integrals in the algebras, we see that
the Casimir operators $C_{(3)}$ and $C_{(5)}$ of the cubic and quintic algebras have the following general form
\begin{align*}
    C_{(3)}   =&  \hat{F}^2 + w_1 \{\hat{X}_1^2,\hat{X}_2\} + w_2\{\hat{X}_1,\hat{X}_2^2\} + w_3\{\hat{X}_1,\hat{X}_2\} + w_4 \hat{X}_1^4 \\
    & + w_5 \hat{X}_1^3 + w_6 \hat{X}_1^2 + w_7 \hat{X}_1 + w_8 \hat{X}_2 + w_9 \hat{X}_2^2 ,\\
    C_{(5)}  = & \hat{G}^2 + \omega_1  \{\hat{Y}_1^3,\hat{Y}_2\} + \omega_2\{\hat{Y}_1^2,\hat{Y}_2\} + \omega_3  \{\hat{Y}_1,\hat{Y}_2^2\} + \omega_4 \{\hat{Y}_1,\hat{Y}_2\}  \\ 
    & + \omega_5 \hat{Y}_2^2 + \omega_6 \hat{Y}_2 +
     \omega_7 \hat{Y}_1^6 +\omega_8 \hat{Y}_1^5+ \omega_9 \hat{Y}_1^4+ \omega_{10} \hat{Y}_1^3+ \omega_{11} \hat{Y}_1^2+ \omega_{12} \hat{Y}_1,
\end{align*} where $w_j$ and $\omega_j$ are coefficients.
%\begin{equation}
%   [C_{(3)},\hat{X}_1] = [C_{(3)},\hat{X}_2] = [C_{(3)},\hat{F}] = 0 \text{ and }     [C_{(5)},\hat{Y}_1] = [C_{(5)},\hat{Y}_2] = [C_{(5)},\hat{K}] = 0. \label{eq:123}
%\end{equation}   

Now using the quadratic commutation relations $\eqref{eq:a}$, we have \begin{align*}
    [C_{(3)},\hat{X}_1]     =& -(u_1+  w_1)\{\hat{F} ,\hat{X}_1^2\}  -( w_3+u_2+w_2 u_1)\{\hat{F} ,\hat{X}_1\} -( w_9+u_3)\{\hat{F},\hat{X}_2\}   \\
    & - (2 u+ w_8 + w_2 u_2)\hat{F} - w_2 \{\hat{F} ,\{\hat{X}_1,\hat{X}_2\}\}  + \hat{X}_1(w_1)\{\hat{X}_1^2,\hat{X}_2\} + \ldots+ \hat{X}_1(w_9)\hat{X}_2^2.
\end{align*} Setting the coefficients to be zero gives
\begin{align*}
    w_1 = - u_1,\quad \text{ } w_2 =  0, \quad\text{ } w_3 = -u_2 , \quad \text{   } w_8 = - 2u,\quad \text{ }w_9 =- u_3.
\end{align*} Similarly, from $[C_{(3)},\hat{X}_2] = 0,$, we have 
\begin{align*}
   0    = & (2w_4 - v_1)\{\hat{F} ,\hat{X}_1^3\}+   (\frac{3w_5}{2}  - v_2) \{\hat{F} ,\hat{X}_1^2\}+ (w_6 - v_3-u_1^2) \{\hat{F} ,\hat{X}_1\} \\ 
     &   +(w_7-u_1 u_2   - 2v) \hat{F}. 
\end{align*} This gives
\begin{align*}
    w_4 = \frac{v_1}{2}, \quad\text{ } w_5 = \frac{2v_2}{3},\quad\text{ } w_6 = v_3 + u_1^2  , \quad\text{ }w_7 = u_1u_2  + 2v.
\end{align*}
This finishes the proof for $C_{(3)}$.

The derivation of $C_{(5)}$ is slightly more complicated. We express $[C_{(5)},\hat{Y}_1]$ and $[C_{(5)},\hat{Y}_2]$ in terms of $\{\hat{G},\hat{Y}_1^n\}$ and $\{\hat{G},\{\hat{Y}_1,\hat{Y}_2\}\}$. By $\cite[\text{Lemma 2}]{MR3205917}$ and quintic commutation relations, we have
\begin{align*}
%& = \alpha \{K,\hat{Y}_1^3\}_a + \beta \{K,\hat{Y}_1^2\}_a +   \delta\{K,\hat{Y}_1\}_a + \epsilon \{K,\hat{Y}_2\} + 2\zeta K- \omega_1 \{K,\hat{Y}_1^3\}_a - \omega_2\{K,\hat{Y}_1^2\}_a \\
%& - \omega_3\left( \{K,\{\hat{Y}_1,\hat{Y}_2\}_a\}_a+ \alpha [\hat{Y}_1^3,\hat{Y}_2] + \beta[\hat{Y}_1^2,\hat{Y}_2] + \delta K \right) - \omega_4[ \hat{Y}_1^2,\hat{Y}_2]\\
%    & + \omega_5\{K,\hat{Y}_2\}_a -\omega_6 K  + \hat{Y}_1(\omega_1)\{\hat{Y}_1^3,\hat{Y}_2\}_a + \ldots + \hat{Y}_1(\omega_{11})\hat{Y}_1^2 +\hat{Y}_1(\omega_{12})\hat{Y}_1 \\
[ C_{(5)},\hat{Y}_1]    = & -(\alpha + \omega_1)\{\hat{G},\hat{Y}_1^3\}-\left(\beta   + \omega_2-  \frac{3\alpha\omega_3}{2}\right)\{\hat{G},\hat{Y}_1^2\} -   (\delta+\omega_3 \beta+ \omega_4)\{\hat{G},\hat{Y}_1\}    \\
& - (\epsilon + \omega_5)\{\hat{G},\hat{Y}_2\}  - \omega_3\{\hat{G},\{\hat{Y}_1,\hat{Y}_2\}\}  +\left[ \left(\frac{\alpha\epsilon  }{2}   -\delta\right) \omega_3- 2\zeta-\omega_6\right]  \hat{G} .
\end{align*} 
Setting the coefficients of $\{\hat{K},\{\hat{Y}_1,\hat{Y}_2\}\} $ and $\{\hat{K},\hat{Y}_2\}$ to be zero, we obtain that $\omega_3 = 0$ and $\omega_5 = - \epsilon$.  Then $[C_{(5)},\hat{Y}_1]$ is reduced to the form 
\begin{align*}
        [ C_{(5)},\hat{Y}_1]  =&  (\alpha - \omega_1)\{\hat{G},\hat{Y}_1^3\} + (\beta   - \omega_2 )\{\hat{G},\hat{Y}_1^2\}+   [\delta-  \omega_4)]\{\hat{G},\hat{Y}_1\}         + (  2\zeta-\omega_6)  \hat{G}.
\end{align*} 
From $[C_{(5)},\hat{Y}_1] = 0$ it follows that the coefficients of $\{\hat{G},\hat{Y}_1^l\} $ are zero for all $ 1 \leq n \leq 3.$ Thus
\begin{align*}
        \omega_1 = -\alpha  , \text{ } \omega_2 = -\beta  , \text{ } \omega_4 =   -\delta  \text{ and } \omega_6 = -2 \zeta .
\end{align*} 
Similarly, after some manipulations we find
  \begin{align*}
  [C_{(5)},\hat{Y}_2] = &   (3\omega_7  -a) \{\hat{G},\hat{Y}_1^5\} +  \left(\frac{5\omega_8  }{2}  - b\right)\{\hat{G},\hat{Y}_1^4\}  - (c+   3\alpha\delta +5\epsilon \omega_7  -2\omega_9  ) \{\hat{G},\hat{Y}_1^3\}\\ 
    & +  \left[\frac{3\alpha}{2} \left(\frac{\alpha \epsilon}{2} - \delta \right)- \beta^2+\frac{3\alpha\delta\epsilon}{2}+ 3 \epsilon^2 \omega_7 - e- \epsilon \omega_9+ \omega_{11} \right]   \{\hat{G},\hat{Y}_1\} \\
  &    - \left(\beta  \delta   +\frac{ \epsilon\omega_8  }{2}-   \frac{1}{2}\omega_{10}   -    3\alpha \beta    +d  \right)\{\hat{G},\hat{Y}_1^2\} \\
    &     +  \left[  \beta  \left(\frac{\alpha \epsilon}{2} - \delta  \right)- \frac{\epsilon\omega_{10}}{2}  + \epsilon^2\omega_8  +  \frac{3\alpha \beta\epsilon}{2}   + (\omega_{12}    - 2z )\right]    \hat{G} .
\end{align*} 
It follows from $[C_{(5)},\hat{Y}_2] =0$ that
\begin{align*}
    & \omega_7 = \frac{a}{3},\quad \text{ } \omega_8 = \frac{2b}{5},\quad\text{ }  \omega_9 = \frac{1}{2} \left( c + \frac{5 \epsilon}{3} + 3 \alpha \delta \right),\quad \text{ } \omega_{10} = 2(\beta(\delta + 3\alpha) + \frac{\epsilon b}{5} -d) \\
    & \omega_{11} = \left(\frac{5}{6} -a \right)\epsilon^2 + e + \frac{\epsilon c}{2} + \beta^2 - \frac{3\alpha}{2}\left(\frac{\alpha \epsilon}{2} - \delta  \right), 
    \quad \omega_{12} = 2z + \beta \delta + \beta\epsilon(\alpha + \delta) - \frac{b\epsilon^2}{5} - \epsilon d 
\end{align*} as required.
\end{proof}

\vskip.1in
we now construct realizations of these algebras in terms of the deformed oscillator algebras (\ref{eq:alg}) and determine their structure functions. After long computations, we obtain the following results.

\vskip.1in
\begin{proposition}
 \label{3.3}
 The realization
\begin{align}
\hat{X}_1=&\sqrt{u_3} \, (\mathcal{N} + \eta),\nonumber\\
\hat{X}_2=&-u_3\, (\mathcal{N} + \eta)^2   - \frac{ u_2 }{\sqrt{u_3}}\, (\mathcal{N} + \eta)  + b^\dagger +  b-\frac{u}{u_3},\label{eq:recu}
\end{align}
where $\eta$ is a constant parameter to be determined, changes the cubic algebra (\ref{eq:a}) to the deformed oscillator algebra (\ref{eq:alg}) with the structure function given by
\begin{align*}
    \Phi(\mathcal{N},\eta)  =& \frac{1}{1-2 u_3} \left\{ C_{(3)}-\frac{u^2}{u_3}+\frac{u u_2}{\sqrt{u_3}}+\sqrt{u_3} v 
     +(\mathcal{N}+\eta)^2 \left(-2 u u_1+2 u_1 u_2 \sqrt{u_3}-u_2^2+(u_2+v_2) u_3^{3/2}-u_3\right) \right.\\
     &\quad +(\mathcal{N}+\eta) \left(2 u u_1-\frac{2 u u_2}{\sqrt{u_3}}-\sqrt{u_3} (u_1 u_2+2 v)+u_2^2+u_3 v_3\right)\\
     &\left.\quad +(\mathcal{N}+\eta)^3 \left(-2 u_1 u_2 \sqrt{u_3}+2 u_1 u_3^2-\frac{2}{3} v_2 u_3^{3/2} + v_1u_3^2 \right) 
     +(\mathcal{N}+\eta)^4 \left(-2 u_1 u_3^2+u_3^3-\frac{1}{2}v_1 u_3^2 \right)\right\}.
\end{align*} 
Note that $\Phi$ is a quartic polynomial of the number operator $\mathcal{N}$.
\end{proposition}

\vskip.1in
\begin{proposition}
 \label{3.4}
 The transformation
\begin{align}
\hat{Y}_1=&\sqrt{\epsilon} (\mathcal{N} + \eta),\nonumber\\
\hat{Y}_2=&-  \alpha \sqrt{\epsilon}  (\mathcal{N} + \eta)^3 -\beta (\mathcal{N} + \eta)^2 - \frac{\delta}{\sqrt{\epsilon}}  (\mathcal{N} + \eta) +b^\dagger+b - \frac{\zeta}{\epsilon},\label{eq:requ}
\end{align}
where $\eta$ is a constant parameter to be determined, maps the quintic algebra (\ref{eq:b}) to the deformed oscillator algebra with the structure function
\begin{align*}
       \Phi(\mathcal{N},\eta)=&\frac{1}{4 \epsilon}C_{(5)}+(\mathcal{N} + \eta)^6 \left(\frac{a \epsilon ^4}{12}-\frac{3 \alpha ^2 \epsilon ^3}{4}\right)+(\mathcal{N} + \eta)^5 \left(\frac{3 \alpha ^2 \epsilon }{4}+\frac{1}{2} \alpha  \beta  \epsilon ^{5/2}-\frac{a \epsilon ^2}{4}+\frac{1}{10} b \epsilon ^{7/2}\right)\\
      & +(\mathcal{N} + \eta)^4 \left(-\frac{1}{4} \alpha  \beta  \sqrt{\epsilon }+\frac{1}{8} \epsilon ^3 \left(3 \alpha  \delta +\frac{5 a \epsilon }{3}+c\right)+\frac{1}{2} \alpha  \delta  \epsilon ^2+\frac{\beta ^2 \epsilon ^2}{4}-\frac{1}{4} b \epsilon ^{3/2}\right)  \\
      &+(\mathcal{N} + \eta)^3 \left(-\frac{1}{4} \alpha  (2 \alpha  \epsilon -\delta )-\frac{3 \alpha  \delta }{4}+\frac{1}{2} \alpha  \zeta  \epsilon ^{3/2}-\frac{\beta ^2}{2}+\frac{1}{2} \beta  \delta  \epsilon ^{3/2}-\frac{c \epsilon }{4}+\frac{1}{4} \omega_{10} \epsilon ^{5/2}\right) \\
      & +(\mathcal{N} + \eta)^2 \left(\frac{\beta  (2 \alpha  \epsilon -\delta )}{4 \sqrt{\epsilon }}-\frac{3 \alpha  \zeta }{4 \sqrt{\epsilon }}-\frac{\beta  \delta }{2 \sqrt{\epsilon }}+\frac{\beta  \zeta  \epsilon }{2}+\frac{\delta ^2 \epsilon }{4}-\frac{d \sqrt{\epsilon }}{4}+\frac{\omega_{11} \epsilon ^2}{4}\right) \\
      &+(\mathcal{N} + \eta) \left(\frac{\delta  (2 \alpha  \epsilon -\delta )}{4 \epsilon }-\frac{\beta  \zeta }{2 \epsilon }+\frac{1}{2} \delta  \zeta  \sqrt{\epsilon }-\frac{e}{4}+\frac{1}{4} \omega_{12} \epsilon ^{3/2}\right)+\frac{\zeta ^2}{4}+\frac{\zeta  (2 \alpha  \epsilon -\delta )}{4 \epsilon ^{3/2}} .
\end{align*} 
The structure function $\Phi(\mathcal{N})$ is a polynomial of $\mathcal{N}$ of degree 6.
\end{proposition}

In the next subsection, we will present new superintegrable systems in 2D Darboux spaces with cubic symmetry algebras.
 
%Recall that the Hamiltonian $\mathcal{H}_0$ for a free particle can be considered as the geodesic flow of a metric $g.$ Suppose that it admits a linear integral. From $\cite{MR675528}$, for some non-zero function $f(x),$ $g$ has the form of $g = f(x) (d^2x + d^2y)$.

\subsection{Superintegrable systems in 2D Darboux spaces with cubic symmetry algebras }
\label{2}

In this subsection we obtain potentials in the 2D Darboux spaces which can be added to the Hamiltonians of the free superintegrable systems studied in $\cite{MR3988021}$ and preserve their superintegrability.  The free systems have only kinetic terms and possess linear and quadratic integrals of motion. We will determine the integrals corresponding to the superintegrable systems with potnetials.

\subsubsection{Darboux space I}

The Hamiltonian of the free system in Darboux space I with separable local coordinates $(x,y)$ studied in $\cite{MR3988021}$ has the form $\mathcal{H}_1=\varphi_1 (x) (\partial_x^2 + \partial_y^2)$, where $\varphi_1 (x)= \frac{1}{\alpha x + \beta}$. This system has linear integral $X_1=\partial_y$ and quadratic integral given by
$$X_2=y \partial_x \partial_y - x \partial_y^2 + \frac{1}{2}\partial_x - \frac{1}{4}\alpha y^2 \varphi_1(x) (\partial_x^2 + \partial_y^2), $$
where $\alpha$ is a constant.

We seek new superintegrable system in Darboux space I with Hamiltonian
$$ \hat{\mathcal{H}}_1=\mathcal{H}_1+V_1(x,y),$$
where $V_1(x,y)$ is potential function,  which preserves the separability of the coordinates and the superintegrability of the original system. Without loss of generality, we assume that the local separable coordinates $(x,y)$ is an orthogonal system. After some computations, we find the allowed potential $V_1$ and the corresponding integrals $\hat{X}_1, \hat{X}_2$. The results are as follows.
\begin{align*}
   &\hat{\mathcal{H}}_1  = \varphi_1(x) (\partial_x^2 + \partial_y^2) + c_1 \varphi_1(x), \\
   &\hat{X}_1  = \partial_y, \qquad \text{ } \hat{X}_2 = y \partial_x \partial_y - x \partial_y^2 + \frac{1}{2}\partial_x - \frac{1}{4}\alpha y^2 \varphi_1(x) (\partial_x^2 + \partial_y^2)  - \frac{1}{4}c_1 \alpha \varphi_1(x)y^2
\end{align*}
where $c_1$ is a constant.

By a direct calculation, we can show that the integrals $\hat{X}_1,\,\hat{X}_2$ form the cubic algebra,
\begin{align}
  [\hat{X}_1,\hat{X}_2] = \hat{F}, \quad\text{ } [\hat{X}_1,\hat{F}] = \frac{\alpha}{2} \hat{\mathcal{H}}_1, \quad\text{ }  [\hat{X}_2,\hat{F}] =   -2 X_1^3 +  \alpha \hat{\mathcal{H}}_1 X_1  - c_1 X_1, \label{eq:h1}
\end{align}  
where explicitly $ \hat{F} = \partial_x \partial_y - \frac{1}{2} \alpha y \varphi_1(x) \left( \partial_x^2 + \partial_y^2 \right)  + \frac{1}{2} c_1 \alpha \varphi_1(x) y $. This cubic algebra is a special case of $\eqref{eq:a}$ in Proposition $\ref{3.1}$ with 
\begin{align*}
  v_1 = -2, \quad\text{ } u_1=u_2 = u_3 = v_2 = v=0, \quad\text{ } u =  \frac{\alpha}{2}\hat{\mathcal{H}}_1, \quad\text{ } v_3 =  \beta \hat{\mathcal{H}}_1-c_1.
\end{align*} 
Then it follows that its Casimir operator is 
$$C_{(3)} = \hat{F}^2 - X_1^4 - \alpha \hat{\mathcal{H}}_1 \hat{X}_2 +(\beta \hat{\mathcal{H}}_1-c_1) X_1^2. $$ 
%the functional independent relation for integrals is \begin{align}
%     \hat{F}^2 + X_1^4 +d X_1^2 -  \alpha \hat{\mathcal{H}}_1  X_2 = 0, \label{eq:c1} \end{align} where $d = c_1 - \beta \hat{\mathcal{H}}_1.$ 
Since $u_3 = 0$ it follows from Proposition $\ref{3.3}$ that this cubic algebra does not have realization in terms of the deformed oscillator algebra.

%Let \begin{align}
%    \hat{X}_1' = \alpha_1 X_1, \text{ } \hat{X}_2' =\beta_1 \hat{X}_2 + \beta_2 X_1 + \beta_3 X_1^2, \text{ } \hat{F}'= [\hat{X}_1',\hat{X}_2'] \label{eq:new}
%\end{align} be a new basis for the cubic algebra. From the commutator relation $\eqref{eq:h1}$, we find that \begin{align*}
%    &[\hat{X}_1',\hat{X}_2'] = \alpha_1 \beta_1 \hat{F} = \hat{F}'; \\
%    & [\hat{X}_1',\hat{F} ] = \frac{\alpha_1^2 \beta_1 \alpha}{2} \hat{\mathcal{H}}_1 ;\\
%    & [\hat{X}_2',\hat{F} ] = -\frac{2 \beta_1}{\alpha_1^3} (\hat{X}_1')^3 + \frac{1}{\alpha_1}\left( \beta_1 \beta \hat{\mathcal{H}}_1 - c_1 \beta_1 + \frac{\beta_3 \alpha}{2} \hat{\mathcal{H}}_1 \right) + \frac{\alpha \beta_2}{2} \hat{\mathcal{H}}_1 
%\end{align*} forms a cubic algebra. See more detailed explanation in $\cite{marquette2022generalized}$. Again from Lemma $\ref{3.3},$ there is no realization for this algebra.

\subsubsection{Darboux space II}

The Hamiltonian of the free superintegrable system in 2D Darboux space II is 
$\mathcal{H}_2=\varphi_2(x) \left(\partial_x^2+\partial_y^2\right)$, where $\varphi_2(x)=\frac{x^2}{a_2-a_1x^2},\; a_1,\,a_2\in \mathbb{R}$. The system possesses the following linear and quadratic integrals of motion,
$$X_1=\partial_y,\quad X_2=2xy\partial_x\partial_y+(y^2-x^2)\partial_y^2+ x\partial_x + y\partial_y + a_1y^2 \mathcal{H}_2.$$

It can be shown that we can add the potential $V_2(x,y)=c_2\,\varphi_2(x)$, where $c_2$ is a real constant, to the free Hamiltonian such that
\begin{align*}
    \hat{\mathcal{H}}_2 = \varphi_2(x) \left(\partial_x^2+\partial_y^2\right)+c_2\,\varphi_2(x)
\end{align*}
is separable and superintegrable in the 2D Darboux space II, with integrals of motion given by
\begin{align*}
&\hat{X}_1 = \partial_y,\\
&\hat{X}_2 = 2 xy \partial_x \partial_y + (y^2 - x^2+1) \partial_y^2 + x\partial_x + y\partial_y + a_1y^2 \mathcal{H}_2 + \frac{a_2 c_2 y^2}{a_2 - a_1x^2}.
\end{align*}

%From the constraints $[\hat{\mathcal{H}}_2,\hat{X}_j]=0$ for all $j =1,2,$ we deduce that \begin{align*}
%&  V_2 = \varphi_2(x) c_2, \text{ } f_2(x,y) = \frac{a_2 c_2 y^2}{a_2 - a_1x^2}
%\end{align*}  
By a direct computation, we find that these integrals obey the cubic commutation relations
\begin{align}
 & [\hat{X}_1,\hat{X}_2] = \hat{F} , \qquad \text{ }  [\hat{X}_1,\hat{F}] =  2 a_1 \hat{\mathcal{H}}_2 + 2 \hat{X}_1^2 + 2 c_2 ,\nonumber\\
 & [\hat{X}_2,\hat{F}] =4\hat{X}_1^3-2\{\hat{X}_1,\hat{X}_2\}+ (2c_2+1-2a_2 \hat{\mathcal{H}}_2) X_1.  \label{eq:h2}
 \end{align} 
%where \begin{align*}
%     \hat{F} = x \partial_x \partial_y + y\partial_y^2 + \frac{1}{2}\partial_y + a_1 y \mathcal{H} + \frac{2 a_2 c_2 y}{a_2 - a_1 x^2}.
% \end{align*} 
The Casimir operator of this cubic algeba is given by
 $$C_{(3)} = \hat{F}^2 - 2\{X_1^2,\hat{X}_2\}_a +2 \hat{X}_1^4+ (c_2+5 - 2 a_2 \hat{\mathcal{H}}_2) X_1^2 - 4(a_1 \hat{\mathcal{H}}_2 +c_2 )\hat{X}_2. $$ 
% and operator identities satisfy \begin{align}
%     \hat{F}^2 -  \frac{4}{3} \{X_1^2,\hat{X}_2\} - 4(a_1 \hat{\mathcal{H}}_2+   c_2) \hat{X}_2+ \left(\frac{11}{3}    - 4 a_2 \hat{\mathcal{H}}_2 \right) X_1^2 - \frac{4}{3} X_1 \hat{X}_2 X_1 +  \frac{2a_1}{3}\hat{\mathcal{H}}_2  + \frac{2c_2}{3}  = 0 \label{eq:19}
% \end{align}  
By Proposition $\ref{3.1}$ and Proposition $\ref{3.3}$, we again find that the cubic algebra has no realization in terms of the deformed oscillator algebra. 

%Using the new basis in $\eqref{eq:new},$  From the commutator relation $\eqref{eq:h1}$, we find that \begin{align*}
%    &[\hat{X}_1',\hat{X}_2'] = \alpha_1 \beta_1 \hat{F} = \hat{F}' ;\\
%    & [\hat{X}_1',\hat{F} ] = \alpha_1^2 \beta_1 \left( 2 a_1\hat{\mathcal{H}}_2 + 2 c_2 \right) + 2 \beta_1 (\hat{X}_1')^2;\\
%    & [\hat{X}_2',\hat{F} ] =-2\beta_1 \{\hat{X}_1',\hat{X}_2'\}_a  + \left((1 - 4 a_2 \hat{\mathcal{H}}_2) \beta_1 + 2 \beta_1 \beta_3(c_2 + a_2 \hat{\mathcal{H}}_2)\right) \hat{X}_1'+ \frac{4 \beta_1 \beta_2}{\alpha_1} (\hat{X}_1')^2 + \frac{4\beta_3 \beta_1}{\alpha_1^2} (\hat{X}_1')^3
%\end{align*} forms a cubic algebra. Again from Lemma $\ref{3.3},$ there is no realization for this algebra. 

\subsubsection{Darboux space III}

In the 2D Darboux space III, the free superintegrable system Hamiltonian and its constants of motion in the separable local coordinates $(u,v)$ are given by \begin{align*}
   &\mathcal{H}_3  = \varphi_3(v) (\partial_u^2 + \partial_v^2),\\
   &X_1 =  \partial_u , \quad \text{ }  X_2   = \frac{1}{2} e^{-v} \left(  \cos u(2 \partial_u^2 + \partial_v) +  \sin u\, (2 \partial_u \partial_v - \partial_u)\right) + \alpha \cos u \mathcal{H}_3 ,
\end{align*} 
where $\varphi_3(v) =  \frac{e^{-v}}{\beta e^v - 2\alpha}$ with $\alpha,\,\beta$ being real constants.

We seek potential of the form $V_3(u,v) = \varphi_3(v) (f_3 (u) + g_3(v))$  such that system in Darboux space III with this potential is superintegrable. We thus expect that $\hat{\mathcal{H}}_3 =\mathcal{H}_3+V_3(u,v)$ possesses linear and quadratic integrals of the form, $\hat{X}_1=X_1,\; \hat{X}_2=X_2+f_3(u,v)$. After some manipulations, we find that $V_3(u,v) = c_3\,\varphi_3(v)$ and $ f_3(u,v) = \frac{ c_3\, \beta e^v \cos (u)}{2 \beta  e^v-4 \alpha } $, where $c_3$ is a constant. That is, we obtain the superintegrable system in Darboux space III with Hamiltonian and integrals given by
\begin{align*}
&\hat{\mathcal{H}}_3=\varphi_3(v) (\partial_u^2 + \partial_v^2) +\frac{c_3}{\beta e^v - 2\alpha},\\
&\hat{X}_1=\partial_u , \\
&\hat{X}_2   = \frac{1}{2} \exp(-v) \left(  \cos u(2 \partial_u^2 + \partial_v) +  \sin u (2 \partial_u \partial_v - \partial_u)\right) + \alpha \cos u \mathcal{H}_3+\frac{ c_3\, \beta e^v \cos (u)}{2 \beta  e^v-4 \alpha }.
\end{align*}
These integrals form the following algebra,
\begin{align}
      [\hat{X}_1,\hat{X}_2] =  \hat{F}, \quad\text{ }   [\hat{X}_1,\hat{F}] =  -\hat{X}_2, \quad \text{ }  [\hat{X}_2,\hat{F}] =  -\beta \hat{\mathcal{H}}_3 \hat{X}_1.\label{eq:h3}
\end{align} 
%is a special case of the generic cubic algebra given in Proposition $\ref{3.1}$ with the coefficients 
%\begin{align*}
%    u_1 = u_2 = v_1 = v_2 = v =u=0, \quad \text{ } u_3 = -1, \quad \text{ } v_3 =  - \beta\hat{\mathcal{H}}_3.
%\end{align*} 
The Casimir operator of this algebra is given by $ C_{(3)}  = \hat{F}^2    - \beta \hat{\mathcal{H}}_3 X_1^2   + \hat{X}_2^2 $. It is interesting that the algebra generated by the above linear and quadratic integrals in the Darboux space III is ``linear" in the generators (though with coefficient involving the Hamiltonian $\hat{\mathcal{H}}_3$).
%and the functional independent relation of integrals is \begin{align}
%     C_{(3)} - \alpha^2 \hat{\mathcal{H}}_3^2 + \left(    \frac{\beta}{4} -  c_3 \alpha  \right)\hat{\mathcal{H}}_3 - \frac{c_3^2}{16}=0.
%\end{align} 
%By Proposition $\ref{3.3}$, 
%\begin{align*}
%    \hat{X}_1   = i (\mathcal{N}  + \eta), \qquad \hat{X}_2 = (\mathcal{N}  + \eta)^2 + b^\dagger + b 
%\end{align*} 
%convert the cubic algebra into the deformed oscillator algebra with structure function $\Phi_3(\mathcal{N},\eta)$ given by
%\begin{align*}
%      \Phi_3(\mathcal{N},\eta) &  = \frac{1}{3}\left( \alpha^2 \hat{\mathcal{H}}_3^2 - \left(    \frac{\beta}{4} -  c_3 \alpha  \right)\hat{\mathcal{H}}_3 + \frac{c_3^2}{16}  + \hat{\mathcal{H}}_3 (\mathcal{N} +\eta)        + \hat{\mathcal{H}}_3 (\mathcal{N} +\eta)^2 -   (\mathcal{N}  + \eta)^4     \right), 
%\end{align*} 
%where $\eta$ is constant parameter to be determined. 

\subsubsection{Darboux space IV}

In terms of separable local coordinates $(u,v)$, the Hamiltonian of the free superintegrable system in 2D Darboux space IV is 
$\mathcal{H}_4=\varphi_4(u) \left(\partial_u^2+\partial_v^2\right)$, where $\varphi_4(u)=\frac{\sin^2u}{\beta - 2 \alpha \cos u}\; \alpha,\,\beta\in \mathbb{R}$. The system possesses the following linear and quadratic integrals of motion,
$$X_1=\partial_v,\quad X_2=\frac{ \exp(v)}{2} \left( \cos u (2 \partial_v^2 - \partial_v) -\sin u(2 \partial_u\partial_v - \partial_v) - 2 \alpha \mathcal{H}_4\right).$$

By analysis similar to previous cases, we find that the system with the Hamiltonian
\begin{align*}
    \hat{\mathcal{H}}_4=\varphi_4(u) \left(\partial_u^2+\partial_v^2\right)+\frac{c_4 }{\beta - 2 \alpha \cos u},
\end{align*} 
where $c_4$ is a constant, is superintegrable with linear and quadratic integrals given by
$$\hat{X}_1=\partial_v,\quad \text{ } \hat{X}_2=\frac{\exp(v)}{2} \left( \cos u (2 \partial_v^2 - \partial_v) -\sin u(2 \partial_u\partial_v - \partial_v) - 2 \alpha \mathcal{H}_4\right)+ \frac{4 c_4 e^{-v}}{\beta - 2 \alpha \cos u}. $$
These integrals form the cubic algebra,
\begin{align}
    [X_1,\hat{X}_2] = \hat{F} , \quad\text{ } [X_1,\hat{F}] = -  \hat{X}_2, \quad\text{ } [\hat{X}_2,\hat{F}] =  4 X_1^3 -2 \beta  \hat{\mathcal{H}}_4 X_1 + \frac{1}{2} \hat{X}_1 - 2 \alpha  c_4 X_1. \label{eq:44}
\end{align}  
The Casimir operator of the algebra is $$C_{(3)} =   - \frac{1}{2} \{\hat{X}_2,\hat{F}\}_a  + \beta \hat{\mathcal{H}}_4 X_1^2 + X_1^4 + (5 + 4 c_4) X_1^2 , $$ which can be expressed as $ C_{(3)} = \hat{\mathcal{H}}_4^2 + \beta \hat{\mathcal{H}}_4 + 4 c_4 $ in terms of the Hamiltonian $ \hat{\mathcal{H}}_4$.  Through the change of basis,
\begin{align*}
    X_1   =   (\mathcal{N}  + \eta), \quad X_2 = -(\mathcal{N}  + \eta)^2 + b^\dagger + b  
\end{align*} the cubic algebra relations become those of the deformed oscillator algebra with structure function \begin{align*}
    \Phi(\mathcal{N},\eta) = (\mathcal{N}+\eta)^4-4 ( +\mathcal{N}+\eta)^3+(\mathcal{N}+\eta)^2-(\mathcal{N}+\eta) \left(-2 \alpha  c_4-2 \beta  E+\frac{1}{2}\right)-4 c_4-E^2-\beta  E.
\end{align*}
Here $\eta$ is a constant which can be determined from the constraints on the structure function.

\section{Conclusions}

We have presented a genuine algebraic analysis for the superintegrable systems in 2D Darboux spaces. The main results in this paper are following. 

The first main result is the construction of the Casimir operators, deformed oscillator algebra realizations and finite-dimensional unirreps for all the $12$ distinct quadratic algebras underlying the 12 superintegrable systems found in the classification of $\cite{MR2023556}$$\cite{MR1878980}$. This allows us to give an algebraic derivation for the energy spectrum of the 12 existing classes of superintegrable systems with quadratic integrals in the 2D Darboux spaces and the determination for the structure functions of the finite-dimensional unitary irreducible representations of the deformed oscillator algebras (corresponding to the quadratic algebras).  As our results demonstrate, superintegrable systems in curved (Darboux) spaces have much richer structures than those in flat spaces. For instance, the structures of energies of the systems and structure functions of the  associated deformed oscillator algebras can be very complicated in the Darboux spaces, and in some cases we have to restrict the model parameter spaces in order to find explicit analytic and closed form solutions. 
 
Another main result of the paper is the construction of generic cubic and quintic algebras, generated by first, quadratic and cubic integrals, their Casimir operators and deformed oscillator algebra realizations. As examples of applications, we obtain four classes of new superintegrable systems with non-trivial potentials and with linear and quadratic integrals in the 2D Darboux spaces, three of which have cubic algebras as their symmetry algebras. 
 
\section*{Acknowledgement}

IM and YZZ were supported by Australian Research Council Future Fellowship FT180100099 and Discovery Project DP190101529, respectively.

\bibliographystyle{unsrt}   % or \bibliographystyle{ieeetr}  to order the references by appearance
\bibliography{bibliography.bib}

\end{document}